\documentclass[12pt]{article}
\usepackage{subcaption}
\usepackage{hyperref}
\usepackage{xurl}
\usepackage{booktabs}
\usepackage{caption}
\usepackage{array}
\usepackage{geometry}
\usepackage{parskip}
\usepackage{hanging}   
\usepackage{authblk}
\usepackage{amsmath,amssymb,placeins}
\usepackage[normalem]{ulem}
\hypersetup{
	colorlinks=true,
	linkcolor=blue,
	urlcolor=blue,
	citecolor=blue
}

\usepackage{amsmath}
\usepackage{amssymb, xcolor, graphicx}
\usepackage{amsthm}
\usepackage{algorithm}
\usepackage{algorithmic}

\usepackage[round]{natbib} 

\newtheorem{theorem}{Theorem}

\newtheorem{assumption}{Assumption}

\title{Scalable Statistical Inference in Stochastic Gradient Descent}
\author[1,2]{Rahul Singh}
\author[3]{Abhinek Shukla}

\affil[1]{ Department of Mathematics, Indian Institute Of Technology Delhi, India}
\affil[2]{Indian Institute of Technology Delhi-Abu Dhabi, Abu Dhabi, UAE}

\affil[3]{ Centre for Biomedical Data Science, Duke-NUS Medical School, Singapore}

\date{Email: \texttt{sirahul@iitd.ac.in}    (RS); \texttt{abhushukla@gmail.com} (AS)}

\begin{document}

\maketitle

\begin{abstract}
Constructing confidence regions for stochastic gradient descent (SGD) ideally requires estimating the asymptotic covariance matrix, a severe computational bottleneck in high dimensions. Traditional cancellation-based batch means methods bypass this estimation but require inverting a sample batch covariance matrix. This introduces strict mathematical degeneracy when the parameter dimension exceeds the number of batches. To address this problem, we utilize equal batch size batch means method and propose a simultaneous, marginal-friendly framework. The proposed marginal statistics has a asymptotic Student's $t$-distribution, and eliminates the matrix inversion step, entirely circumventing high-dimensional degeneracy. To achieve valid simultaneous coverage, we present an algorithm utilizing wild bootstrap samples drawn from a statistic as a function of only the diagonals of the variance-covariance estimator, and to further incorporate the contribution of cross-dependencies, we introduce an efficient Quasi-Monte Carlo procedure utilizing a $t$-copula approximation. Additionally, we integrate a Lugsail variance estimator to aggressively correct finite-sample bias and under-coverage. The proposed methodology delivers interpretable, simultaneous hyper-rectangular confidence regions that are statistically robust, memory-efficient, and strictly scalable for high-dimensional inference. The theoretical results are supported by extensive numerical simulation analysis through various aspects of dimension, number of batches and error structure.

\end{abstract}

\section{Introduction}
Stochastic Gradient Descent (SGD) is a ubiquitous optimization tool for large-scale data, originating from the seminal work of \citet{robbins1951}. Modern data sets benefit from SGD's computational efficiency and online compatibility, leading to its growing popularity \citep[see, e.g., ][]{bottou2010,bottou,sgd2017}.
    
Consider the data arise from $\Pi$, a probability distribution on $\mathbb{R}^{r}$, denoted by $\zeta \sim \Pi$. In a model fitting paradigm, a function $f:\mathbb{R}^{d} \times \mathbb{R}^{r} \to \mathbb{R}$ typically measures empirical loss for estimating a parameter $\theta$, having observed the data $\zeta$. The expected loss is represented by $F(\theta)=\mathbb{E}_{\zeta\sim\Pi}\left[f(\theta,\zeta)\right]$. The  parameter of interest is $\theta^*\in\mathbb{R}^d$ where
\begin{align*}
\theta^*=\arg\min_{\theta\in\mathbb{R}^d}F(\theta).
\end{align*}
The goal is to estimate $\theta^*$ using data $\zeta_i \overset{\text{iid}}{\sim} \Pi$ for $i = 1, \dots, n$, where iid refers to independently and identically distributed. Computing $\theta^*$ naturally involves a gradient based technique. If $F(\theta)$ is unavailable, it can be approximated by replacing it with the empirical loss, $n^{-1} \sum_{i=1}^{n} f(\theta, \zeta_i)$.  When data are online or calculating the whole gradient vector is expensive, an unbiased estimate is used instead. This yields several stochastic gradient methods. Let $\nabla f(\theta,\zeta)$ be the gradient vector of $f(\theta,\zeta)$ with respect to $\theta$, $\eta_i > 0$ be a learning rate, and $\theta_0$ be the SGD process's starting point. The SGD $i^{\text{th}}$ iterate is:
\begin{align*}
\theta_i=\theta_{i-1}-\eta_i\nabla f(\theta_{i-1},\zeta_i)\,, \qquad \text{ for } i = 1,2,\ldots\,\,\,.
\end{align*} 
Although approximations are used in the optimization, SGD estimates of $\theta$ can have good statistical properties \citep{fabian,ruppert1988,polyak1992}, especially when $\eta_i$ is decreasing at a certain rate and the estimate used is averaged SGD (ASGD):
\[
\widehat \theta_n:=n^{-1}\sum_{i=1}^{n}\theta_i\,.
\] 
Point estimates of $\theta^*$ alone are insufficient. \cite{polyak1992} has been crucial in developing a statistical inference framework for $\widehat{\theta}_n$. Define $S := \mathbb{E}_\Pi\left([\nabla f(\theta^*,\zeta)][\nabla f(\theta^*,\zeta)]^\top\right)$ and $A := \nabla^2 F(\theta^*)$ to represent the Hessian of $F(\theta)$ evaluated at $\theta$. When derivative and expectation are interchangeable, $\mathbb{E}_{\Pi} \left[\nabla f(\theta^*,\zeta) \right] = \nabla F(\theta^*) = 0$. \cite{polyak1992} found that if $F$ is convex with a Lipschitz gradient and $\eta_i=\eta i^{-\alpha}$ with $\alpha\in(0.5,1)$, then $\widehat \theta_n$ provides a consistent estimate of $\theta^*$, and under some additional conditions,
\begin{align}
\label{eq:pj_normal}
\sqrt{n}(\widehat \theta_n-\theta^*)\xrightarrow{\texttt{d}} N(0,\Sigma)\text{ as }n\to\infty, \text{ where } \Sigma=A^{-1}SA^{-1}.
\end{align}

To conduct a comprehensive end-to-end study, a practitioner should not only estimate $\theta^*$ but also evaluate the estimator's quality by estimating $\Sigma$, utilizing estimators of $\Sigma$ for inference, and providing uncertainty estimates for predictions. Consequently, statistical inference regarding model parameters represents a progressive step toward the reliable execution of machine learning algorithms. Despite the substantial body of literature addressing the convergence properties of the ASGD estimator and its derivatives \citep{zhang2004,nemi2009,agarwal2012}, estimators for $\Sigma$ have only been recently introduced \citep{fang2018,chen2020,zhu2021,singh2025}.

The estimation of $\Sigma$ facilitates an understanding of the accuracy of the estimation of $\theta^*$, especially in lower-dimensional contexts. Furthermore, estimators of $\Sigma$ can thereafter be utilized to produce estimators of the variability of functions of $\widehat{\theta}_n$, facilitating uncertainty quantification for forecasts. We examine this second aspect with further scrutiny. 
Moreover, by employing the asymptotic normality in \eqref{eq:pj_normal} together with reliable estimators of $\Sigma$, one can perform conventional multivariate hypothesis testing and construct confidence regions.

To the best of our knowledge, all available estimators of $\Sigma$ assume that the number of batches and the batch size both diverge together; see \cite{chen2020,singh2025}. Furthermore, for simultaneous inference, a standard assumption is that the estimate of $\Sigma$ is positive definite, which naturally requires the number of batches to exceed the dimension $d$. However, in many situations, this may not be true, e.g., when the data size is small or the variables are dependent. Current estimators do not provide a solution for inference in this scenario.  {This article aims to address this problem}.

The aim of the paper is two fold:
\begin{enumerate}
    \item to study marginal-friendly adaptation of equal batch size based inference performed in \cite{zhu2021}, along with uplifting the restriction on number of batches to two batches;
    \item to integrate finite sample bias adjustment of \cite{vats2022} and significantly enhance the multivariate coverage.  
\end{enumerate}

\citet{zhu2021} proposed a cancellation-based batch means method in order to estimate $\Sigma$. However, their joint test statistic requires inverting the sample-batch covariance matrix, which requires the number of batches to exceed $d$. Consequently, for large $d$, $m$ must be impractically large, forcing overly small batch sizes that violate the asymptotic Gaussian assumptions of the batch means. To resolve this and enable valid high-dimensional inference, we extend the work of \citet{zhu2021} by developing a simultaneous marginal-friendly framework utilizing Equal Batch-Size (EBS) means. By shifting from a joint Mahalanobis inversion to simultaneous marginal projections, we eliminate the matrix inversion step entirely, thereby completely bypassing the $m > d$ degeneracy requirement.

The remainder of this paper is organized as follows. Section \ref{sec:theory} establishes the asymptotic marginal distribution of the proposed test statistic under standard regularity assumptions, with the integration of the Lugsail bias correction detailed in Section \ref{subsec:lugsail}. Section \ref{sec:algorithms} introduces the proposed Wild Bootstrap and $t$-copula methodologies used to construct the simultaneous confidence hyper-rectangles. Section \ref{sec:simulation} presents an extensive empirical simulation study evaluating the framework across varying dimensionalities, sample sizes, covariate dependency structures, and heavy-tailed error distributions. Finally, Section \ref{sec:conclusion} concludes the paper. Technical proofs of the main theoretical results are deferred to the Appendix.


\section{Theoretical Analysis}
\label{sec:theory}
We start with two widely accepted assumptions for SGD.
\begin{assumption}[On $F$]\label{ass1}
(i) $F({\theta})$ is differentiable and strongly convex with parameter $c > 0$; (ii) $\nabla F({\theta})$ is Lipschitz continuous with parameter $L > 0$; (iii) there exists a constant $C_2 > 0$ such that $\|\nabla F({\theta}) - \nabla^2 F({\theta}^*)({\theta} - {\theta}^*)\| \le C_2 \|{\theta} - {\theta}^*\|^2$; and (iv) $\nabla^2 F({\theta}^*)$ exists.
\end{assumption}

\begin{assumption} \label{ass2}
The sequence $\{\xi_t=\nabla F(\theta_{t-1}) - \nabla f(\theta_{t-1}, \zeta_t)\}_{t \ge 1}$ forms a martingale difference sequence with respect to the filtration $\mathcal{F} = \{\mathcal{F}_t\}_{t \ge 1}$ generated by $\{\zeta_t\}_{t \ge 1}$, and satisfies: (i) Around ${\theta}^*$, $\mathbb{E}[\xi_t \xi_t^\top \mid \mathcal{F}_{t-1}] = U + r(\Delta_{t-1})$ for some positive definite matrix $U$, and there exist constants $S_1, S_2 > 0$ such that for any $x \in \mathbb{R}^d$, $ \|r(x)\| \le S_1\|x\| + S_2\|x\|^2$; (ii) There exists a constant $M \in (0, \infty)$ such that $\|\xi_t\| \le M$ almost surely, for all $t \ge 1$.
\end{assumption}

Under Assumptions \ref{ass1} and \ref{ass2}, \citet{polyak1992} established that
\begin{equation}
\sqrt{n}(\widehat{\theta}_n - \theta^*) \Rightarrow \mathcal{N}(0, \Sigma) \text{ as } n \to \infty
\end{equation} 
where the asymptotic covariance matrix is $\Sigma = A^{-1} S A^{-1}$.

Let the SGD sequence of length $n$ be divided into $m$ equal non-overlapping batches of size $b = \lfloor n/m \rfloor$. Without any loss of generality, we assume that $n/m$ is an integer. Then the  empirical  covariance estimator is
\begin{equation*}
S_m(n) = \frac{1}{m-1} \sum_{i=1}^m (\widetilde\theta_i - \widehat{\theta}_n)(\widetilde\theta_i - \widehat{\theta}_T)^\top,\quad \text{where} \quad
\widetilde\theta_i = \frac{1}{b} \sum_{t = (i-1)b + 1}^{ib} \theta_t.
\end{equation*}
We denote the $j$-th dimension of the ASGD estimator by $\widehat{\theta}_{n, j}$, and its corresponding empirical variance from the diagonal of $S_m(n)$ by $[S_m(n)]_{jj}$. 
We define a statistic for dimension $j$ as
\begin{equation}\label{test:statistic}
t_{n,j} = \frac{\hat{\theta}_{n,j} - \theta^*_j}{\sqrt{[S_m(n)]_{jj} / m}} = \frac{\sqrt{m}(\hat{\theta}_{n,j} - \theta^*_j)}{\sqrt{[S_m(n)]_{jj}}}.
\end{equation}

\bigskip
\begin{theorem}\label{thm1}
Under Assumptions \ref{ass1} and \ref{ass2}, for any $m \ge 2$, for any $d$, $t_{n,j}$ converges in distribution to Student's t-distribution with $m-1$ degrees of freedom.
\end{theorem}

The proof of Theorem \ref{thm1} is deferred to the Appendix \ref{appendix1}. Theorem \ref{thm1} represents a critical departure from the framework of \citet{zhu2021}. In \citet{zhu2021}, the test statistic relies on the quadratic form $(\hat{\theta}_n - \theta^*)^\top S_m(n)^{-1} (\hat{\theta}_n - \theta^*)$, which intrinsically requires the sample batch covariance matrix $S_m(n)$ to be strictly invertible. 
This imposes a hard constraint that the number of batches must strictly exceed the parameter dimension, $m > d$.
In highly parameterized machine learning settings, $m > d$ may not be feasible. In such a scenario, Theorem \ref{thm1} can be utilized. The formulation relies exclusively on the scalar diagonal elements $[S_m(n)]_{jj}$. This completely circumvents the $m > d$ degeneracy limit. Consequently, valid and interpretable marginal inference can be conducted using a small, stable number of batches (e.g., $m \ge 2$), entirely independent of the ambient dimension $d$.

\subsection{Integrating Lugsail Bias Correction}
\label{subsec:lugsail}
\citet{zhu2021} observe that when $n$ is small and $m$ is large, small batch sizes cause the batch means to deviate from their asymptotic Gaussian distribution, introducing pre-limit bias. For reliable finite-sample inference, bias correction is crucial. \citet{singh2025} demonstrated that Equal Batch-Sizes (EBS) are uniquely suited for Lugsail bias correction. 

The standard Lugsail estimator corrects first-order variance bias via the linear combination $\hat{\Sigma}_{Lug} = 2\hat{\Sigma}_{2b} - \hat{\Sigma}_b$, where $b = n/m$ is the batch size. In our cancellation framework, $S_m(n)$ estimates ${m}\Sigma/n$, so $\hat{\Sigma}_b = {n} S_m(n)/m$. Assuming $m$ is even, doubling the batch size to $2b$ halves the number of batches, yielding $\hat{\Sigma}_{2b} = {2n} S_{m/2}(n)/m$. Therefore,
\begin{equation*}
\hat{\Sigma}_{Lug} = 2 \left( \frac{2n}{m} S_{m/2}(n) \right) - \left( \frac{n}{m} S_m(n) \right) = \frac{n}{m} \big( 4 S_{m/2}(n) - S_m(n) \big)
\end{equation*}
Factoring out the shared ${n}/{m}$ scalar to retain the scale-free properties of the cancellation framework, the unscaled Lugsail batch covariance matrix is given by:
\begin{equation*}
S_{Lug}(n,m) = 4 S_{m/2}(n) - S_m(n).
\end{equation*} 
By replacing $[S_m(n)]_{jj}$ with $[S_{Lug}(n,m)]_{jj}$ in \eqref{test:statistic}, we significantly reduce the systematic variance bias common in finite-sample SGD. Moreover, this framework maintains strict memory efficiency and equal batch distributions; for more details, see \cite{singh2025}.

\section{Implementation Schemes}
\label{sec:algorithms}

In finite samples, the exact distribution of \eqref{test:statistic} may deviate significantly from its asymptotic limit, rendering standard inference unreliable. To address this, we introduce two finite-sample correction schemes. Both approaches seamlessly accommodate either the standard Equal Batch-Size (EBS) or the Lugsail bias-corrected covariance estimators. Hereafter, we let $\widehat{\Sigma}$ generically denote either estimator for the true asymptotic covariance $\Sigma$, explicitly defined as $\widehat{\Sigma} = \frac{n}{m} S_m(n)$ for EBS and $\widehat{\Sigma} = \frac{n}{m} S_{Lug}(n,m)$ for Lugsail.

\subsection{Wild--Bootstrap--based Implementation}
We approximate the sampling distribution using a Wild Bootstrap \citep{wildboot}. Unlike the naive bootstrap, this method preserves the sequential order of the batch means. It generates bootstrap samples by perturbing the centered batch means, their deviations from the grand mean, with independent, zero-mean random multipliers. Crucially, this non-parametric technique inherently captures the underlying dimensional dependency structure without requiring the estimation of an explicit correlation matrix. The complete Wild Bootstrap procedure is detailed in Algorithm \ref{alg:wild_bootstrap}.

\begin{algorithm}[h!]
\caption{Wild Bootstrap for Simultaneous Inference}
\label{alg:wild_bootstrap}
\begin{algorithmic}[1]
\REQUIRE Batch means $\{\widetilde{\theta}_i\}_{i=1}^m$, ASGD estimator $\widehat{\theta}_n$, covariance estimate $\widehat{\Sigma}$, Bootstrap iterations $B$, Target confidence level $1-\alpha$.
\FOR{$i = 1$ \TO $m$}
    \STATE $Y_i \leftarrow \widetilde{\theta}_i - \widehat\theta_n$ \quad \textit{(Compute the $d$-dimensional centered deviation for each batch)}
\ENDFOR
\FOR{$k = 1$ \TO $B$}
    \STATE Generate $m$ iid centered random variables $w_1^{(k)}, \dots, w_m^{(k)}$. We took two distributions: Rademacher $w_i \in \{-1, 1\}$ with equal probability, and Standard Normal.
    \STATE Construct the perturbed bootstrap error vector:
     \quad $E^{(k)} \leftarrow \frac{1}{m} \sum_{i=1}^m w_i^{(k)} Y_i$
    \STATE Compute the marginal Lugsail-scaled statistic for each dimension $j \in \{1, \dots, d\}$:
     \quad $t_j^{(k)} \leftarrow \frac{E_j^{(k)}}{\sqrt{[\widehat{\Sigma}]_{jj} / n}}$
    \STATE Extract the maximum absolute statistic across all dimensions to act as a union bound:
     \quad $M^{(k)} \leftarrow \max_{1 \le j \le d} |t_j^{(k)}|$
\ENDFOR
\STATE Set $z^*$ to be the empirical $(1-\alpha)$-quantile of the sorted list $\{M^{(1)}, \dots, M^{(B)}\}$.
\RETURN Marginal-friendly Simultaneous Confidence Region: 
\begin{align}\label{simult:region}
 C_n(z^*) = \prod_{j=1}^d \left[ \widehat\theta_{n, j} - z^* \sqrt{\frac{[\widehat{\Sigma}]_{jj}}{n}}, \quad \widehat\theta_{n, j} + z^* \sqrt{\frac{[\widehat{\Sigma}]_{jj}}{n}} \right] 
 \end{align}
\end{algorithmic}
\end{algorithm}

Algorithm \ref{alg:wild_bootstrap} avoids the intractable derivation of a closed-form generalized $\chi^2$ ratio. Instead, under the Gaussian limit theorem, the perturbed errors $E^{(k)}$ naturally approximate the true asymptotic distribution of the ASGD estimator. Because each $d$-dimensional batch deviation $Y_i$ is multiplied by a single scalar weight $w_i^{(k)}$, the procedure perfectly preserves the empirical cross-covariance structure across all $d$ parameters. Furthermore, evaluating the maximum absolute statistic, $\max_{j} |t_j^{(k)}|$, within each bootstrap iteration inherently adjusts for simultaneous inference. 

\subsection{The $t$-Copula Approximation}
Let $\mathbf{t}_n = (t_{n, 1}, \dots, t_{n, d})^\top$ and $R \in \mathbb{R}^{d \times d}$ denotes the asymptotic correlation matrix of the ASGD estimators, i.e., $R_{jk} = {\Sigma_{jk}}/{\sqrt{\Sigma_{jj}\Sigma_{kk}}}$. From Theorem \ref{thm1}, we have that $\mathbf{t}_n \Rightarrow \mathbf{t}^*$ where each marginal limit follows a Student's t-distribution, $t^*_j \sim t_{m-1}$. Note that the components of $\mathbf{t}^*$ can be expressed as $t^*_j = {\sqrt{m} \bar{Z}_j}/{\sqrt{[S_Z]_{jj}}}$
where $Z_1, \dots, Z_m \sim \mathcal{N}(0, R)$ are independent multivariate normal vectors, $\bar{Z}$ is their sample mean, and $S_Z$ is their sample covariance matrix. While deriving a joint distribution for $\mathbf{t}^*$ is analytically difficult, its dependency structure is governed by the true correlation matrix $R$. We will use this insight to approximate a simultaneous region of the form \eqref{simult:region}.

To compute the simultaneous critical value $z^*$, we approximate this dependency structure using a Student's $t$-copula. By modeling the joint distribution as a standard multivariate Student's $t$-distribution, $\mathcal{T}_{m-1}(0, \hat{R}_n)$, we achieve two critical statistical properties: (i) the exact marginal distribution is preserved; (ii) it inherently captures the tail dependencies and empirical cross-correlations embedded in $\hat{R}_n$. This provides a conservative approximation of the joint confidence bounds without requiring the complex simulation of the full sequence of empirical covariance matrices. The complete steps are provided in Algorithm \ref{alg:t-copula}.

\begin{algorithm}[h!]
\caption{Simultaneous Region via $t$-Copula}
\begin{algorithmic}[1] \label{alg:t-copula}
\REQUIRE 
Batch means $\{\widetilde{\theta}_i\}_{i=1}^m$, ASGD estimator $\widehat{\theta}_n$, covariance estimate $\widehat{\Sigma}$, Target confidence level $1-\alpha$.
\STATE Extract empirical correlation matrix $\hat{R}_n$ from covariance estimator $\widehat{\Sigma}$

\STATE Generate $N$ independent vector samples: 
 \quad $\mathbf{t}^{(k)} \sim \mathcal{T}_{m-1}(0, \hat{R}_n)$
\FOR{$k = 1$ \TO $N$}
    \STATE Extract the maximum absolute statistic across all dimensions:
    \STATE \quad $M^{(k)} \leftarrow \max_{1 \le j \le d} |t^{(k)}_j|$
\ENDFOR

\STATE Set $z^*$ as the empirical $(1-\alpha)$-quantile of the simulated maximums $\{M^{(1)}, \dots, M^{(N)}\}$
\RETURN Simultaneous Confidence Region: 
\[ C_n(z^*) = \prod_{j=1}^d \left[ \hat{\theta}_{n, j} - z^* \sqrt{\frac{[\widehat{\Sigma}]_{jj}}{n}}, \quad \hat{\theta}_{n, j} + z^* \sqrt{\frac{[\widehat{\Sigma}]_{jj}}{n}} \right] \]
\end{algorithmic}
\end{algorithm}

In Algorithm \ref{alg:t-copula}, ``mvtnorm'' R-statistical package can be utilized to sample from the multivariate $t$-distribution, allowing us to compute the joint statistics in a single computationally efficient step. Using this $t$-copula approximation, we implement a efficient Quasi-Monte Carlo (QMC) procedure \citep{robertson2021} to derive the simultaneous hyper-rectangular region.

\section{Simulation Study} \label{sec:simulation}

We evaluate our proposed marginal framework using a linear regression model, $y_i = x_i^\top \theta^* + \epsilon_i$, with $x_i \overset{\text{iid}}{\sim} \mathcal{N}(0, A)$ and independent errors $\epsilon_i \overset{\text{iid}}{\sim} \mathcal{N}(0,1)$. The true parameter vector $\theta^* \in \mathbb{R}^d$ is chosen as an equidistant grid on $(0,1)$, and the objective is the standard squared loss and standard absolute deviation loss. Following \citet{singh2025}, we evaluate three structures for $A$: identity ($A = I_d$), Toeplitz ($A_{i,j} = \rho^{|i-j|}$), and equicorrelation ($A_{i,j} = \rho$ for $i \ne j$ and $1$ otherwise).

Simulations are conducted across dimensionalities $d \in \{20, 30, 50\}$ and sample sizes ranges from $n = 10^4$ to $ n = 5 \times 10^5$, utilizing an SGD learning rate of $\gamma_t = 0.5 t^{-0.51}$. To approximate simultaneous confidence regions, we use $5\times10^4$ iterations for both the $t$-copula and Wild Bootstrap methods. We test batch allocations of $m \in \{10, 100\}$ to analyze Lugsail variance stability. Crucially, when $m \le d$ (e.g., $m=10$), the batch covariance matrix $S_m(n)$ is strictly singular. While this degeneracy structurally invalidates the joint $F$-statistic of \citet{zhu2021}, our marginal procedure completely bypasses the matrix inversion, enabling valid simultaneous inference in these constrained regimes.

The simultaneous confidence regions are computed using the Wild Bootstrap (Algorithm \ref{alg:wild_bootstrap}) and the $t$-copula approximation (Algorithm \ref{alg:t-copula}). All results are based on 5000 replicates. In the plots, we denote the $t$-copula combined with the standard Equal Batch-Size estimator as T EBS, and with the Lugsail bias correction as T LUG. Similarly, the corresponding Wild Bootstrap implementations are denoted as Wild EBS and Wild LUG.

The simulation sections are divided into five parts with the aim of each part is discussed as below:
\begin{enumerate}
    \item Section 4.1 studies the multivariate coverage and interval length behavior of the proposed marginal-friendly hyper-rectangles, implemented via Algorithms \ref{alg:t-copula} and \ref{alg:wild_bootstrap}, for $A$ equal to the identity matrix and $\epsilon_i$ standard normal.
    \item Section 4.2 examines Algorithm \ref{alg:wild_bootstrap} in more depth, comparing discrete (Rademacher) and continuous (normal) weights for generating bootstrap samples of the statistic under consideration, again with $A$ equal to the identity matrix and $\epsilon_i$ normally distributed.
    \item To understand the optimal choice of the number of batches, Section 4.3 presents extensive numerical simulations highlighting the impact of this choice on multivariate coverage and on the (logarithmic) volume of the resulting hyper-rectangle.
    \item Since the marginal distribution of the statistics considered in Theorem \ref{thm1} is available, Section 4.4 compares the proposed marginal-friendly hyper-rectangles with those obtained from a Bonferroni correction based on this marginal distribution. As the Bonferroni-corrected rectangles performed satisfactorily for independent regressors, we also extend the simulations to the equicorrelated and Toeplitz regressor structures.
    \item Section 4.5 further demonstrates the scope of our simulations by incorporating heavy-tailed error structures for $\epsilon_i$ in the linear model, considering two cases: (i) $t$-distributed errors with three degrees of freedom, and (ii) least absolute deviation (LAD) errors for the same linear model.
\end{enumerate}

\subsection{Results under number of batches equal to half of dimension}
Figures \ref{fig:mult_covg}, \ref{fig:se_intv}, \ref{fig:max_intv}, and \ref{fig:vol_rect} present the empirical coverage rates, standard errors of the interval lengths, maximum interval lengths, and log-volumes of the resulting hyperrectangular regions for $ d=20$ and $50$. For the Wild Bootstrap procedure, the plotted results utilize standard normal random multipliers; analogous implementations using Rademacher  multipliers yielded slight  under-performance and are omitted for brevity.

\begin{figure}[htbp!]
	\centering
	
	\begin{subfigure}[htbp!]{0.45\linewidth}
		\centering
		\includegraphics[width=\linewidth]{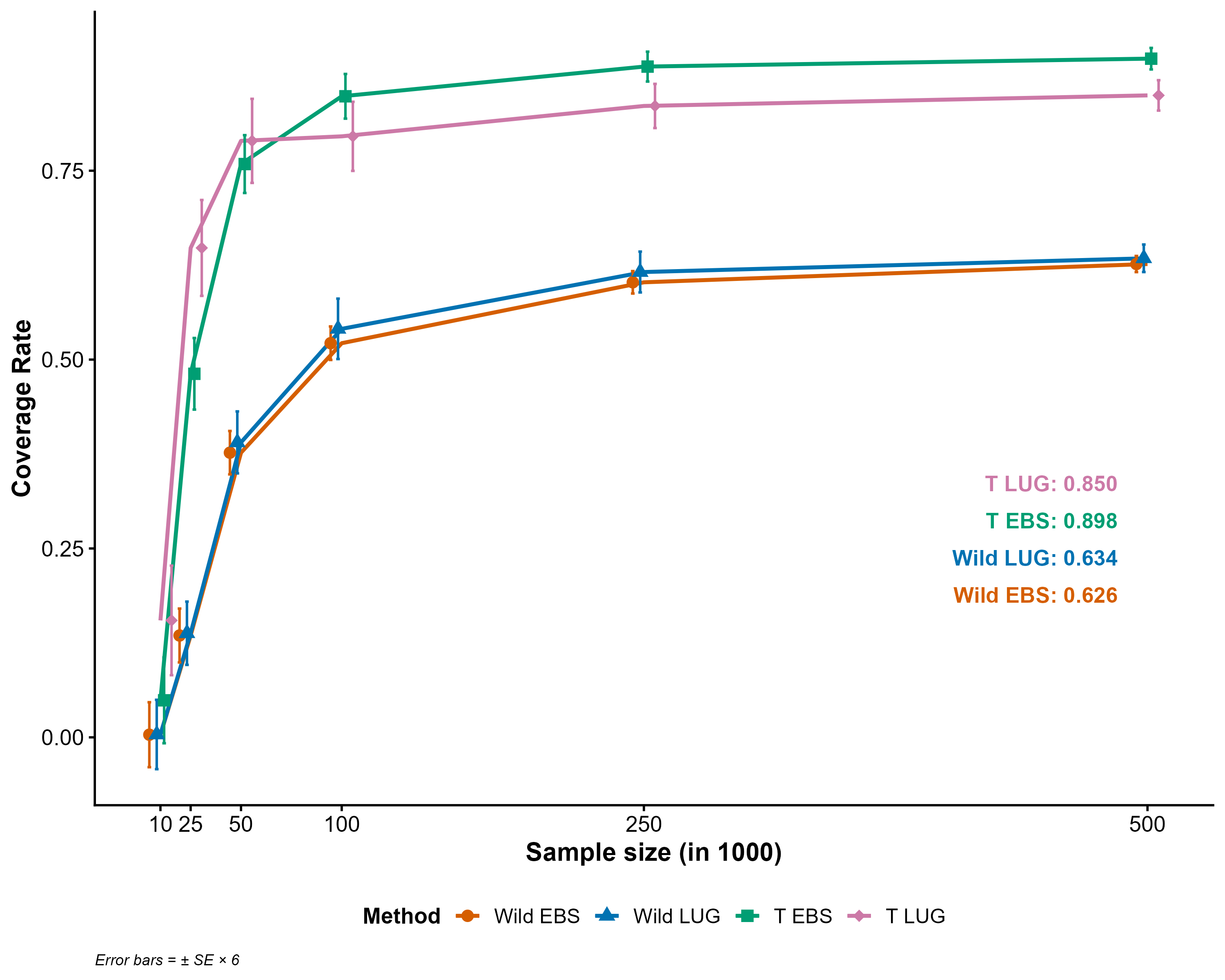}
		\caption{Dimension 20, number of batches 10}
		\label{fig:lin_mult_covg_dim20}
	\end{subfigure}
	\hfill
	\begin{subfigure}[htbp!]{0.45\linewidth}
		\centering
		\includegraphics[width=\linewidth]{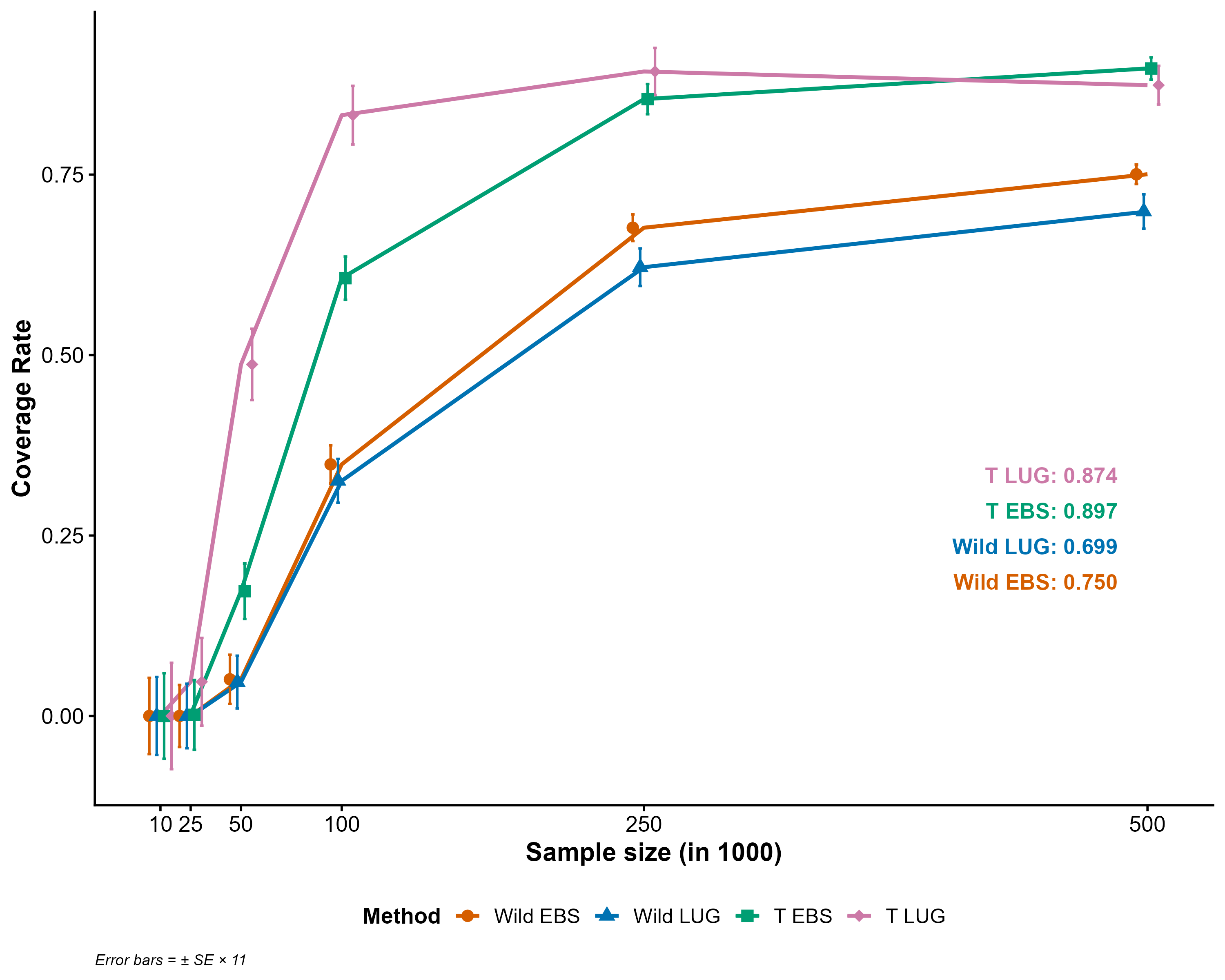}
		\caption{Dimension 50, number of batches 25}
		\label{fig:lin_mult_covg_dim50}
	\end{subfigure}
	
	\caption{{Multivariate coverage based on simultaneous marginal friendly implementation Here Algorithm \ref{alg:wild_bootstrap} is implemented with normal weights.}}
	\label{fig:mult_covg}
\end{figure}

\begin{figure}[htbp!]
	\centering
	
	\begin{subfigure}[htbp!]{0.45\linewidth}
		\centering
		\includegraphics[width=\linewidth]{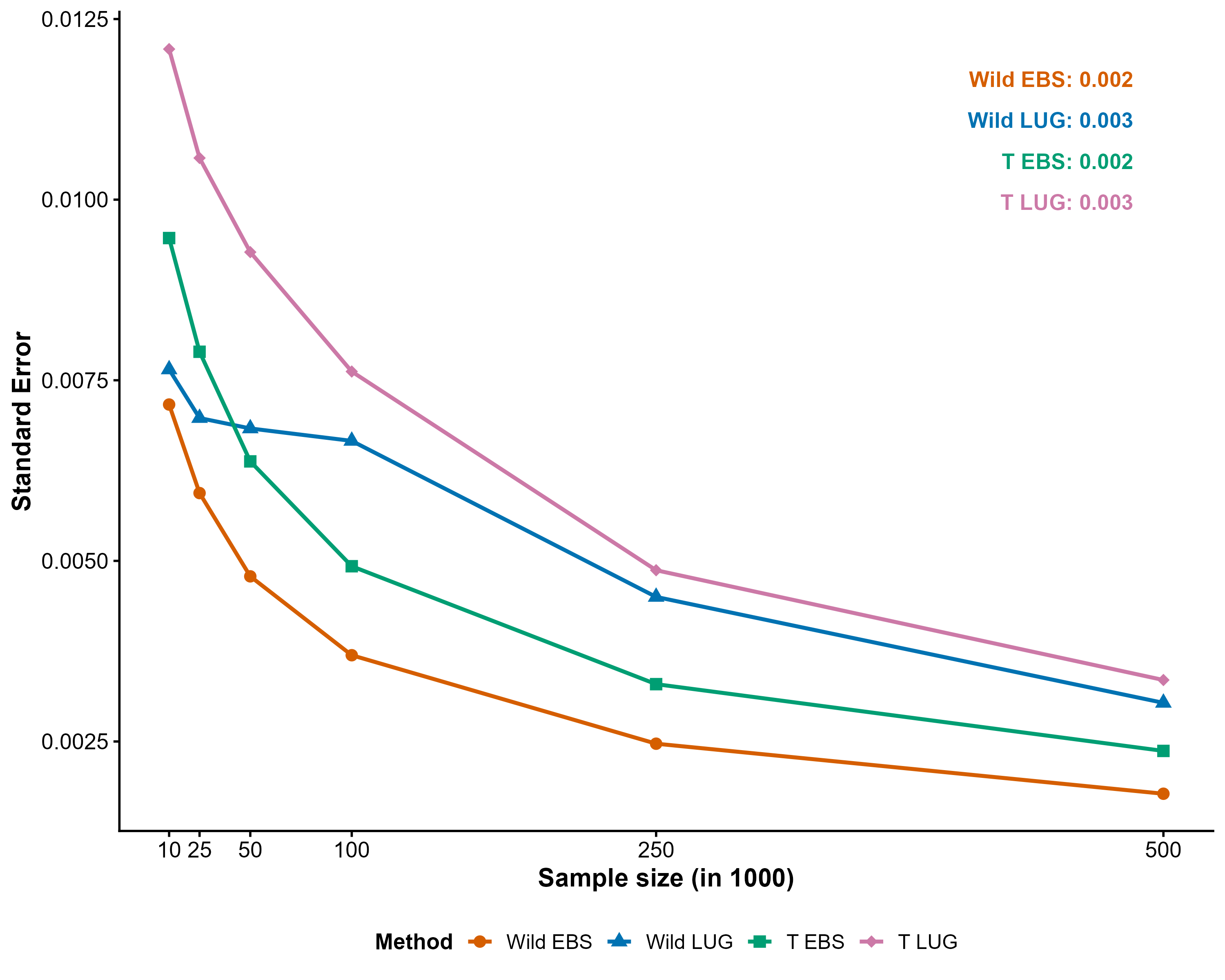}
		\caption{Dimension 20, number of batches 10}
		\label{fig:lin_se_dim20}
	\end{subfigure}
	\hfill
	\begin{subfigure}[htbp!]{0.45\linewidth}
		\centering
		\includegraphics[width=\linewidth]{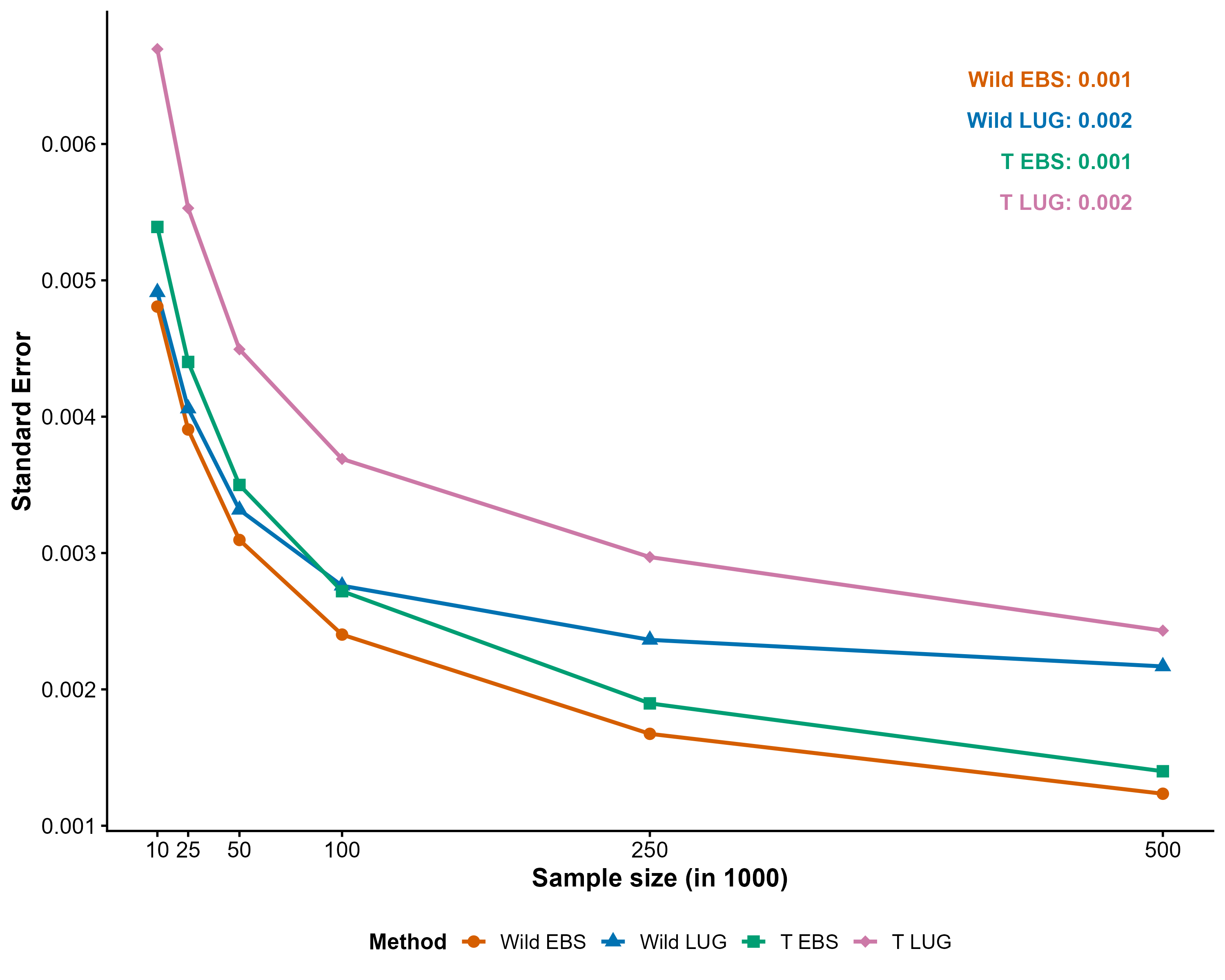}
		\caption{Dimension 50, number of batches 25}
		\label{fig:lin_se_dim50}
	\end{subfigure}
	\caption{{Standard error of interval length across different dimensions,  Here Algorithm \ref{alg:wild_bootstrap} is implemented with normal weights.}}
	\label{fig:se_intv}
\end{figure}

\begin{figure}[htbp!]
	\centering
	
	\begin{subfigure}[htbp!]{0.45\linewidth}
		\centering
		\includegraphics[width=\linewidth]{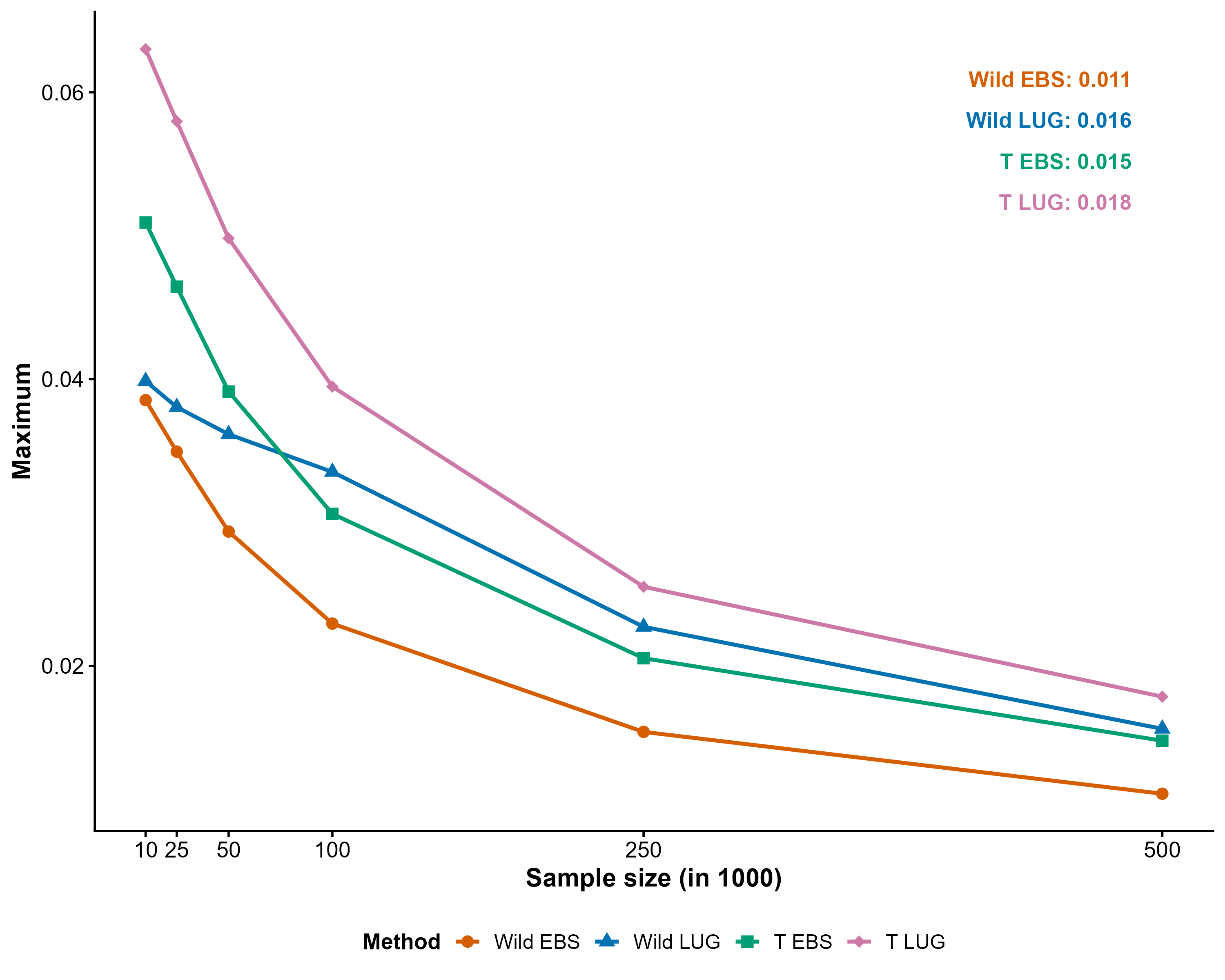}
		\caption{Dimension 20, number of batches 10}
		\label{fig:lin_max_dim20}
	\end{subfigure}
	\hfill
	\begin{subfigure}[htbp!]{0.45\linewidth}
		\centering
		\includegraphics[width=\linewidth]{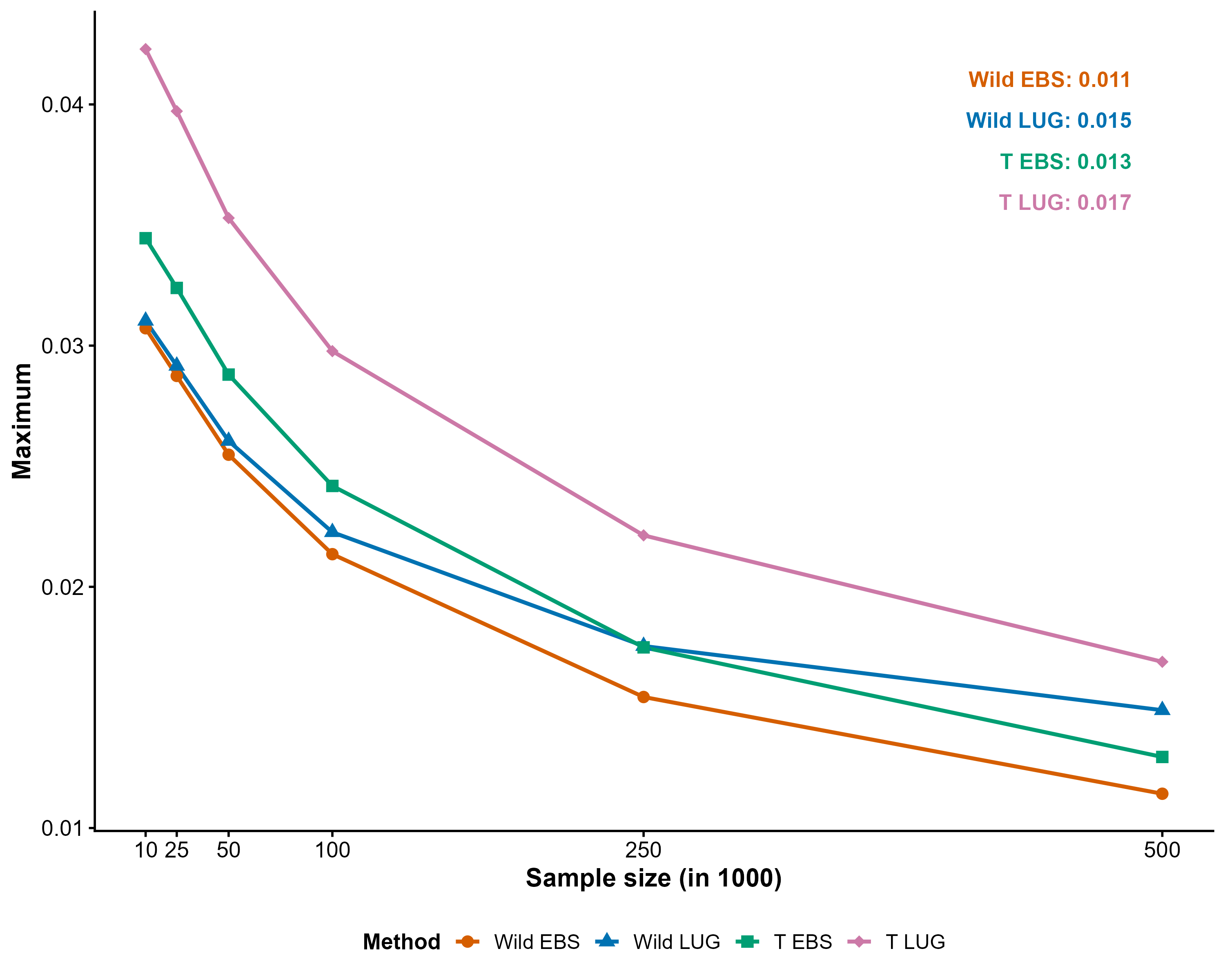}
		\caption{Dimension 50, number of batches 25}
		\label{fig:lin_max_dim50}
	\end{subfigure}
	
	\caption{{Maximum of interval length across different dimensions, Here Algorithm \ref{alg:wild_bootstrap} is implemented with normal weights.}}
	\label{fig:max_intv}
\end{figure}

\begin{figure}[htbp!]
	\centering
	
	\begin{subfigure}[htbp!]{0.45\linewidth}
		\centering
		\includegraphics[width=\linewidth]{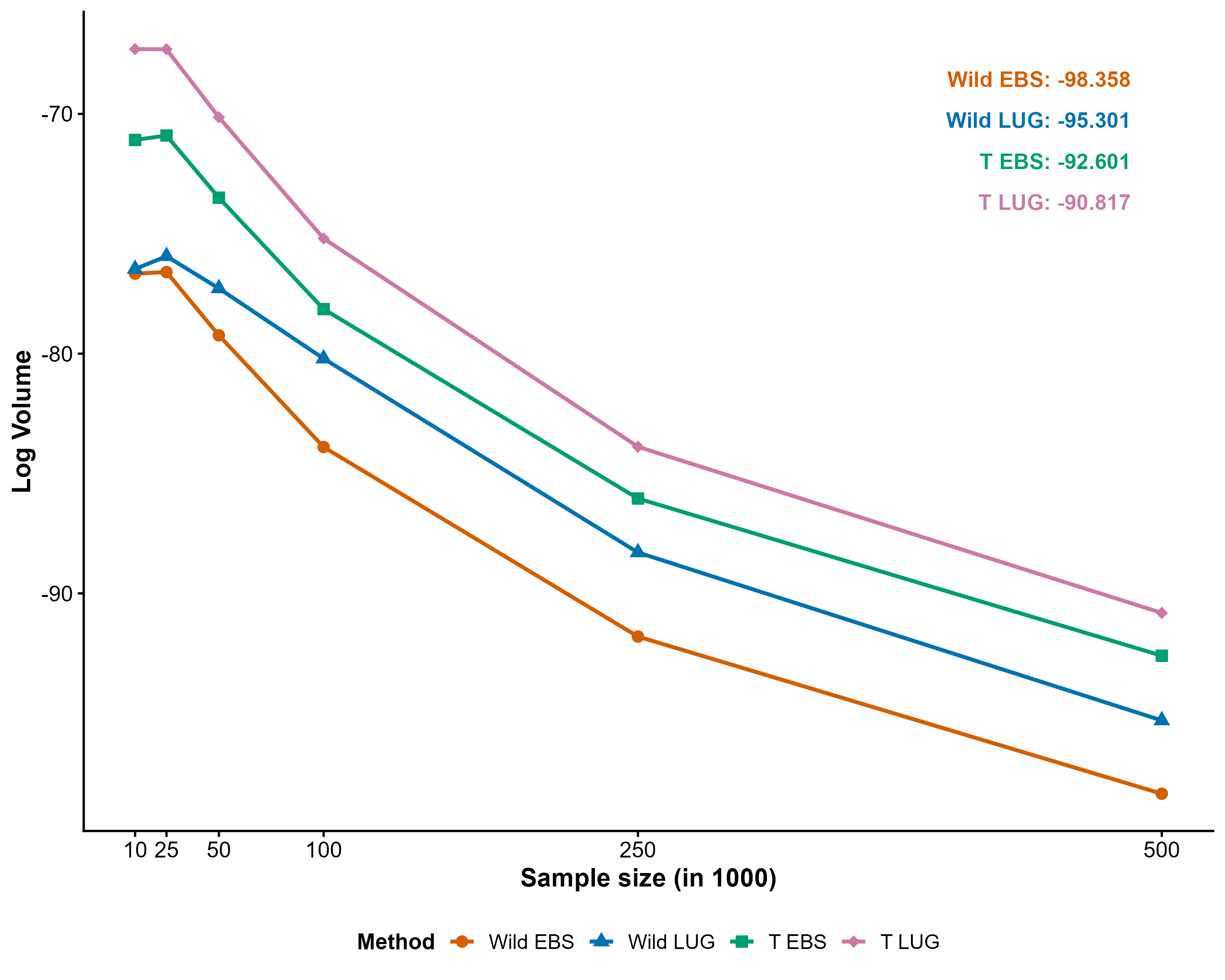}
		\caption{Dimension 20, number of batches 10}
		\label{fig:lin_vol_dim20}
	\end{subfigure}
	\hfill
	\begin{subfigure}[htbp!]{0.45\linewidth}
		\centering
		\includegraphics[width=\linewidth]{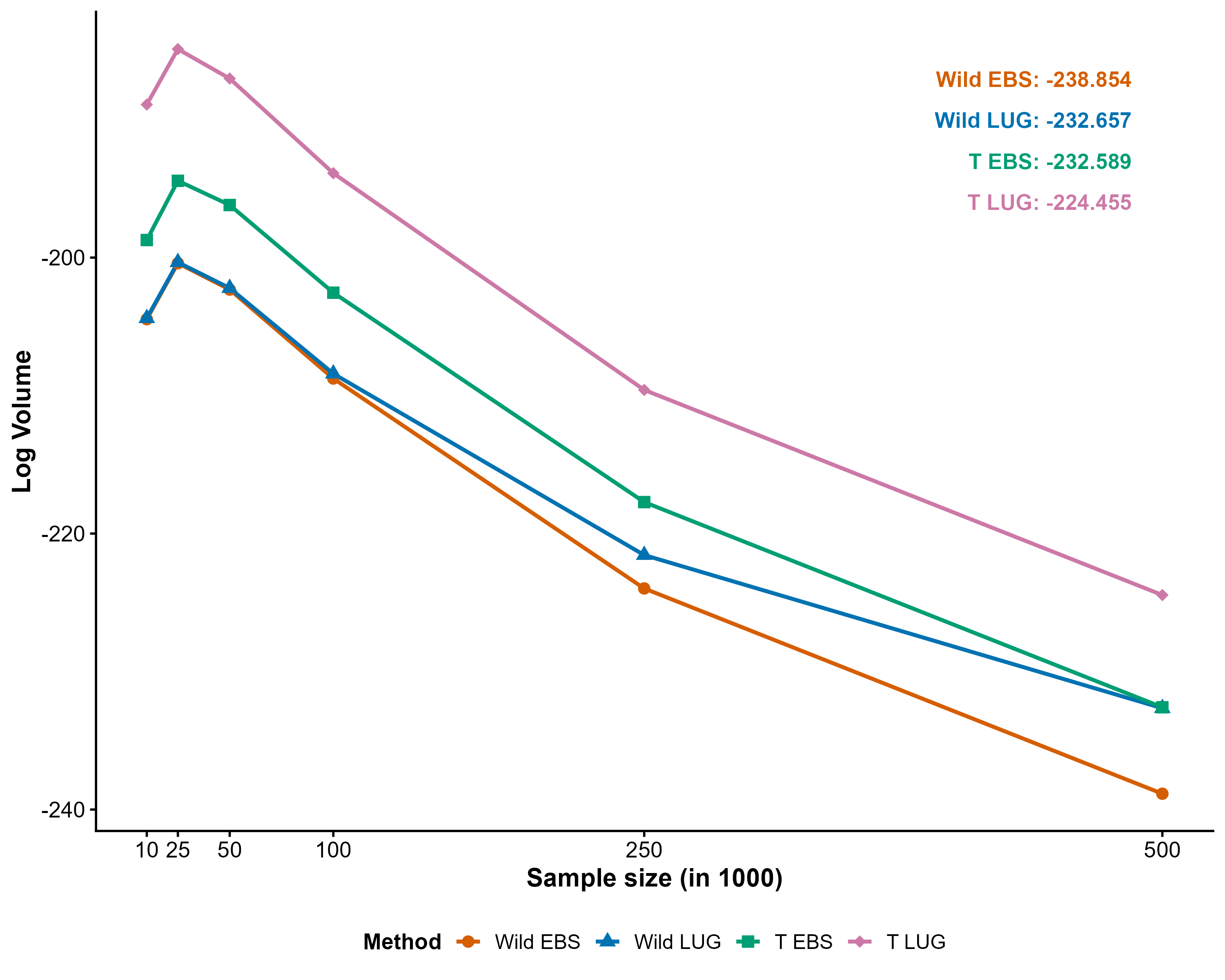}
		\caption{Dimension 50, number of batches 25}
		\label{fig:lin_vol_dim50}
	\end{subfigure}
	
	\caption{{Volume (logarithm) of hyper-rectangle, Here Algorithm \ref{alg:wild_bootstrap} is implemented with normal weights.}}
	\label{fig:vol_rect}
\end{figure}

Both the Wild Bootstrap and the $t$-copula successfully provide valid simultaneous inference even when the number of batches is half of the dimension. As expected, all evaluated interval metrics (standard error,  maximum length, and log-volume) exhibit monotonic decay as the sample size $n$ increases for a fixed $m$ and $d$. Moreover, the $t$-copula consistently outperforms the Wild Bootstrap across all scenarios by explicitly leveraging the full empirical variance-covariance structure. 
\subsection{Rademacher versus normal weights in wild--bootstrap }
When $d>m$, the choice of multiplier distribution in the Wild Bootstrap significantly impacts finite-sample behavior. As illustrated in Figure \ref{fig:rade_norm_comparison}, where $d=20$ and $m=10$, continuous $\mathcal{N}(0,1)$ weights substantially outperform discrete Rademacher weights by a significant margin, achieving strictly higher empirical coverage rates alongside correspondingly larger confidence hyper-rectangle volumes. Furthermore, employing continuous normal multipliers bridges the gap between the two proposed inference strategies, bringing the empirical performance of the Wild Bootstrap much closer to the exact joint covariance modeling achieved by the $t$-copula approximation (Algorithm \ref{alg:t-copula}).

\begin{figure}[htbp!]
    \centering

    \begin{subfigure}[htbp!]{0.48\linewidth}
        \centering
        \includegraphics[width=\linewidth]{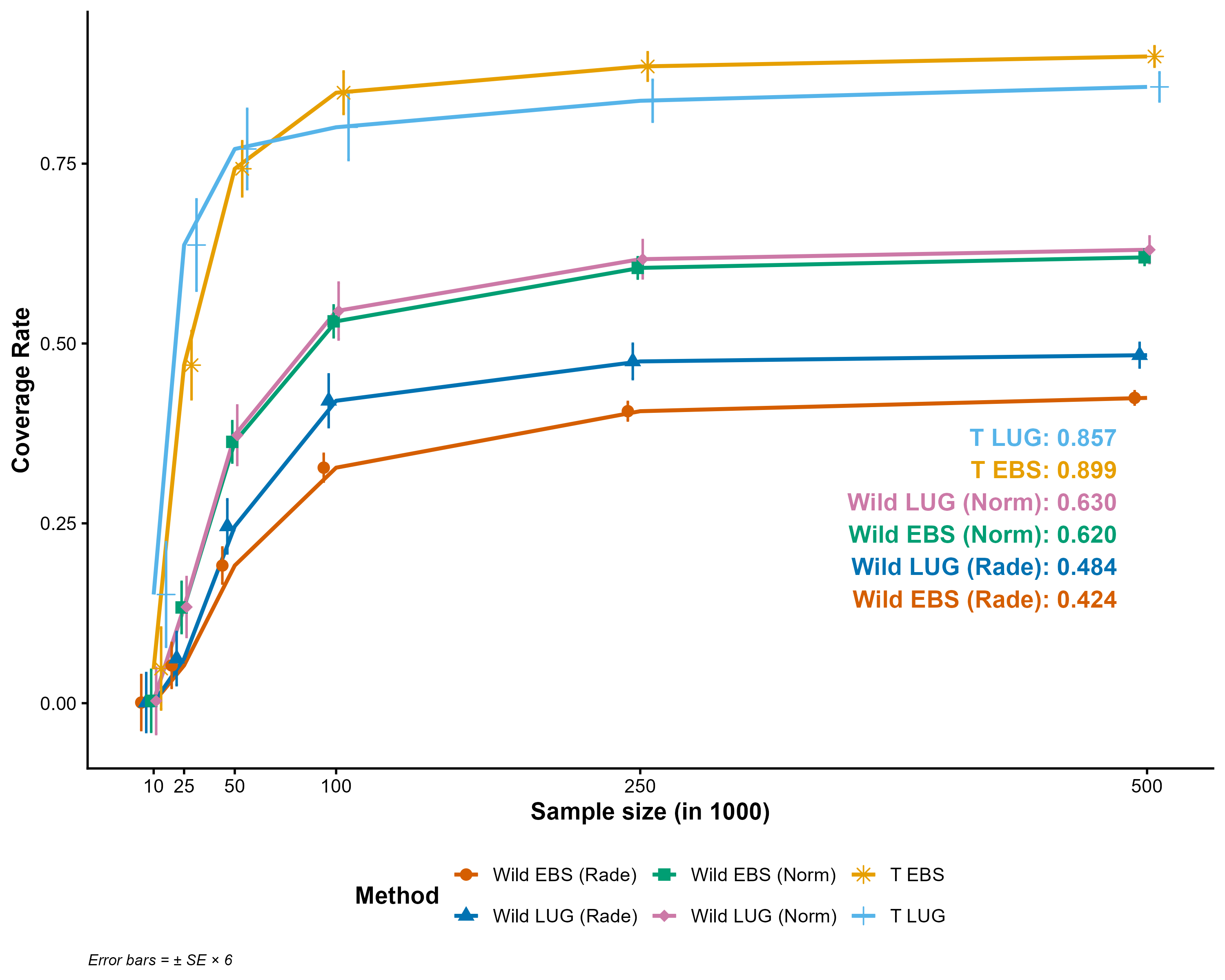}
        \caption{Multivariate coverage}
        \label{fig:sub_covg_rade_norm}
    \end{subfigure}
    \hfill
    \begin{subfigure}[htbp!]{0.48\linewidth}
        \centering
        \includegraphics[width=\linewidth]{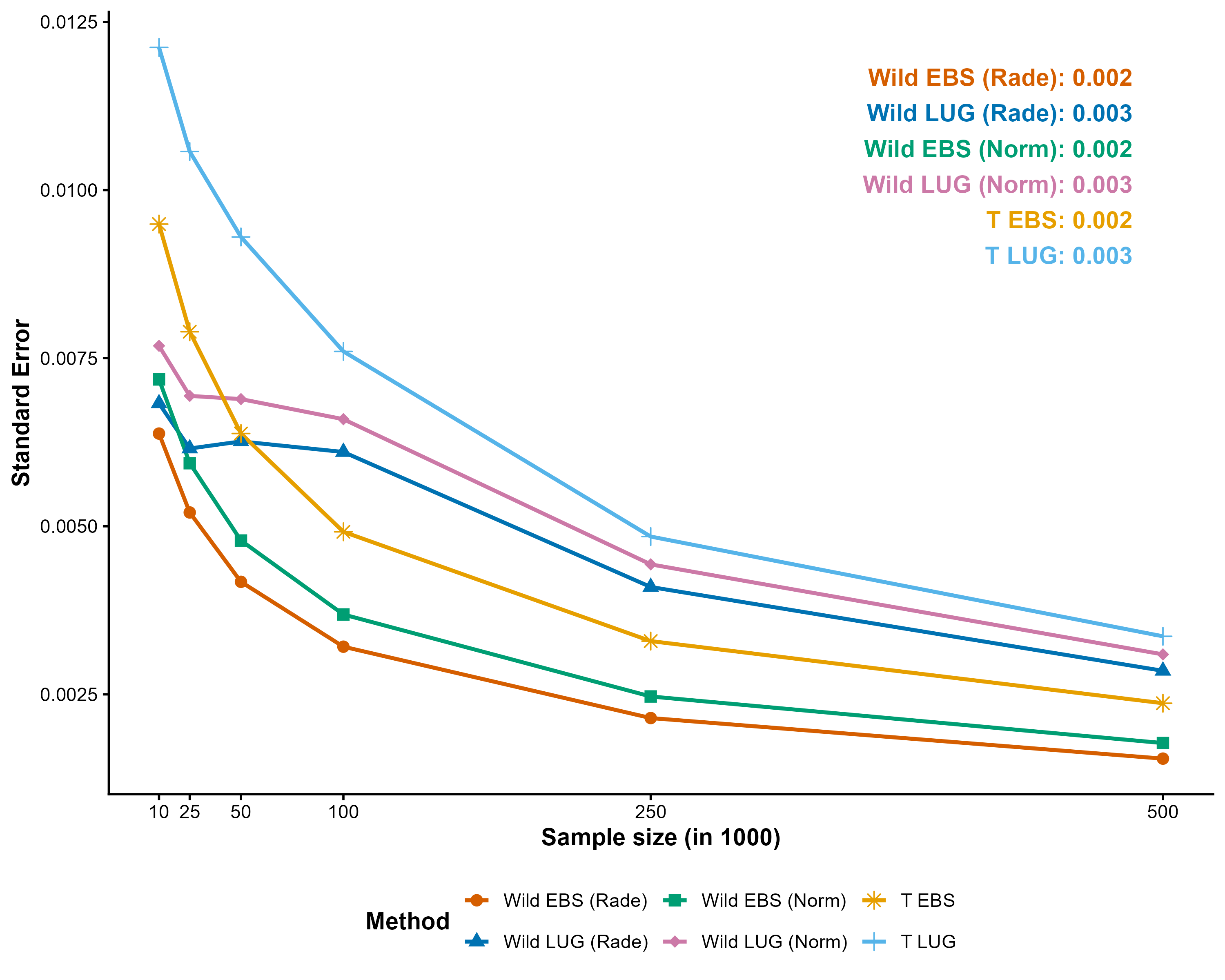}
        \caption{Standard error}
        \label{fig:sub_se_rade_norm}
    \end{subfigure}

    \vspace{0.5em}

    \begin{subfigure}[htbp!]{0.48\linewidth}
        \centering
        \includegraphics[width=\linewidth]{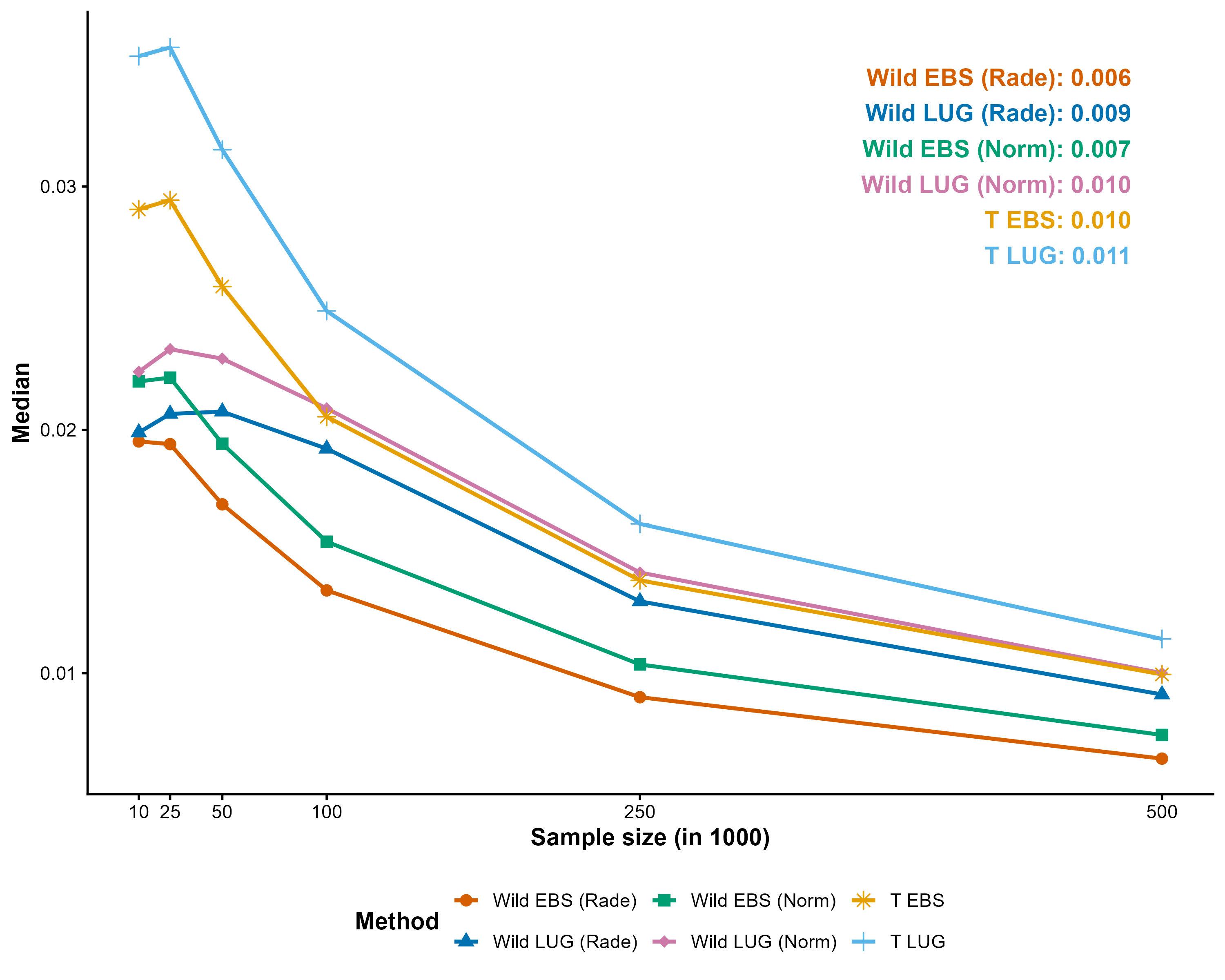}
        \caption{Median interval length}
        \label{fig:sub_med_rade_norm}
    \end{subfigure}
    \hfill
    \begin{subfigure}[htbp!]{0.48\linewidth}
        \centering
        \includegraphics[width=\linewidth]{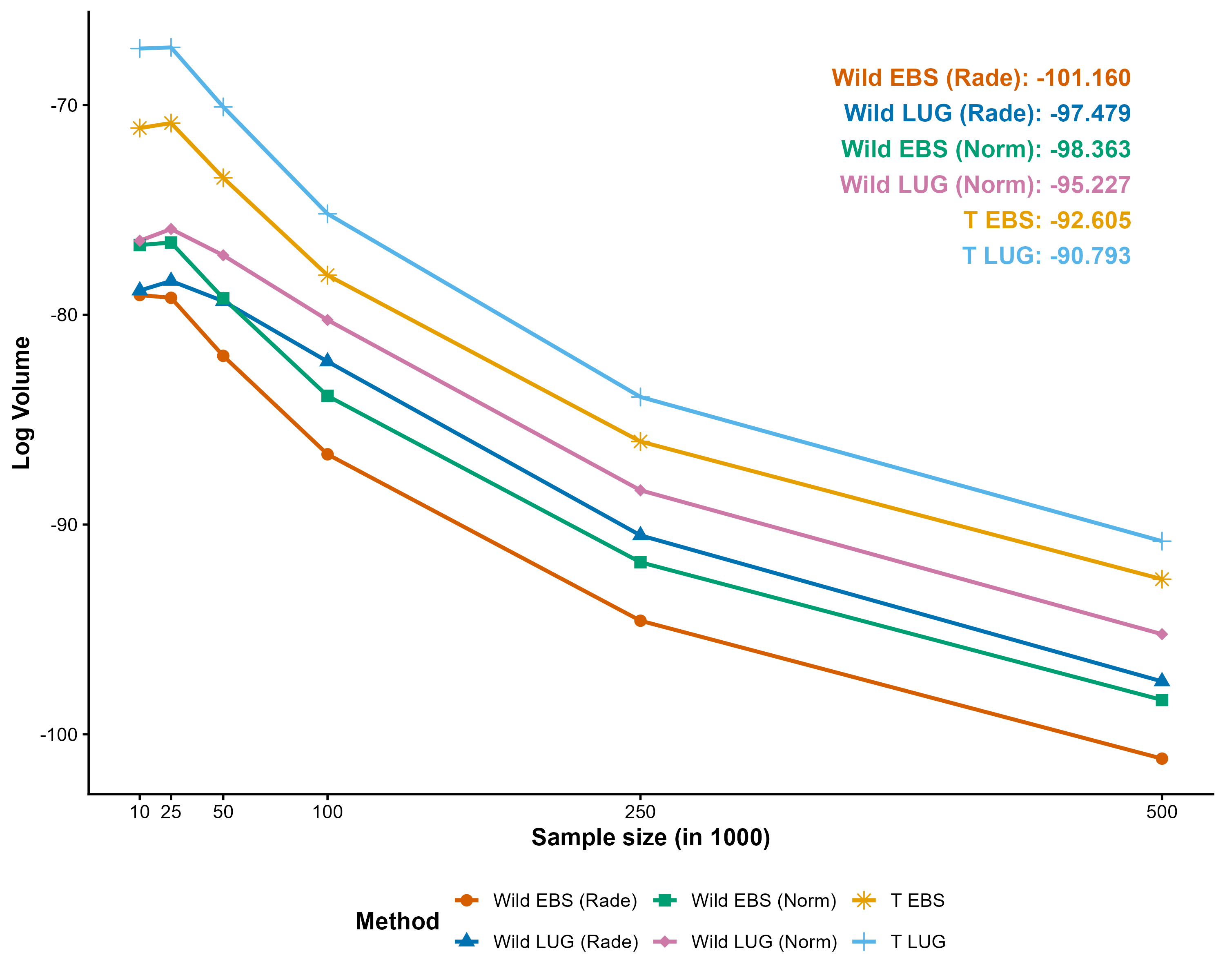}
        \caption{Volume of confidence hyper-rectangle (logarithm)}
        \label{fig:sub_vol_rade_norm}
    \end{subfigure}

    \caption{{Comparison of simultaneous marginal friendly implementation for Wild Bootstrap based on Rademacher and normal weights. Here $d = 20, m = 10$.}}
    \label{fig:rade_norm_comparison}
\end{figure}

\subsection{Sensitivity to the Number of Batches}
To evaluate the empirical effect of the batch quantity $m$, we fix the sample size at $n = 970200$ to ensure the batch size $n/m$ remains an integer across all tested configurations. \citet{zhu2021} heuristically recommend setting $m \in [15, 30]$ for $d \le 10$ and $m \in [d+5, d+10]$ for higher dimensions. Our simulations reveal dimension-dependent behavior. For moderate dimensions ($d=20,~ 30$), the multivariate coverage monotonically increases with larger $m$ (Figure \ref{covg_vary_batch}). Conversely, for $d=50$, the coverage rate begins to deteriorate once $m$ exceeds $d+10$, structurally corroborating the upper-bound threshold suggested by \citet{zhu2021}. Finally, we observe that the average interval length consistently decreases as $m$ increases, successfully preserving the relative performance ordering among the estimators established in our earlier sample size analysis (Figure \ref{vol_rect_vary_batch}).

\begin{figure}[htbp!]
	\centering
	\begin{subfigure}[htbp!]{0.32\linewidth}
		\centering
		\includegraphics[width=\linewidth]{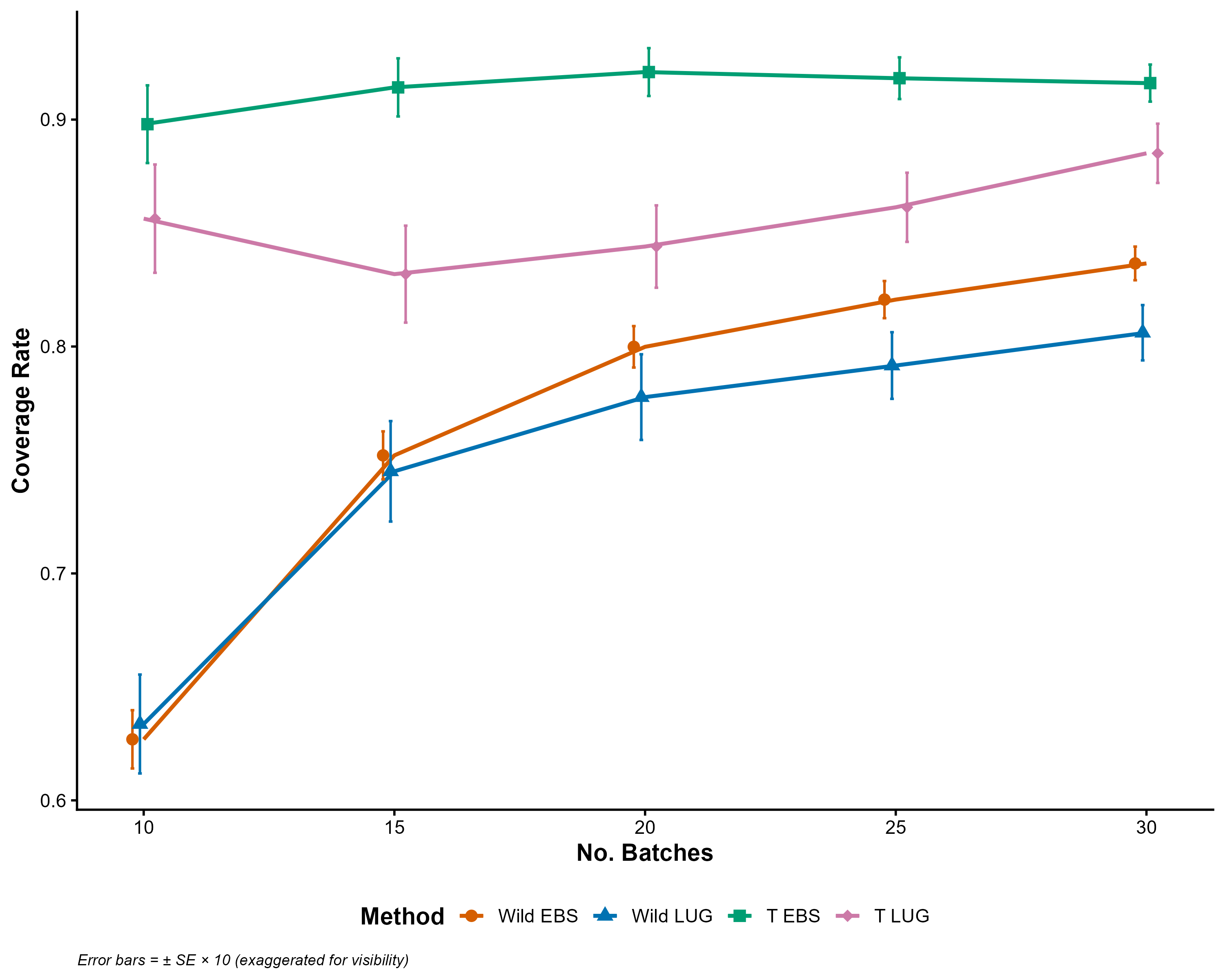}
		\caption{Dimension 20, number of batches 10 to 30.}
		\label{covg_20_batch10_30}
	\end{subfigure}
	\hfill
	\begin{subfigure}[htbp!]{0.32\linewidth}
		\centering
		\includegraphics[width=\linewidth]{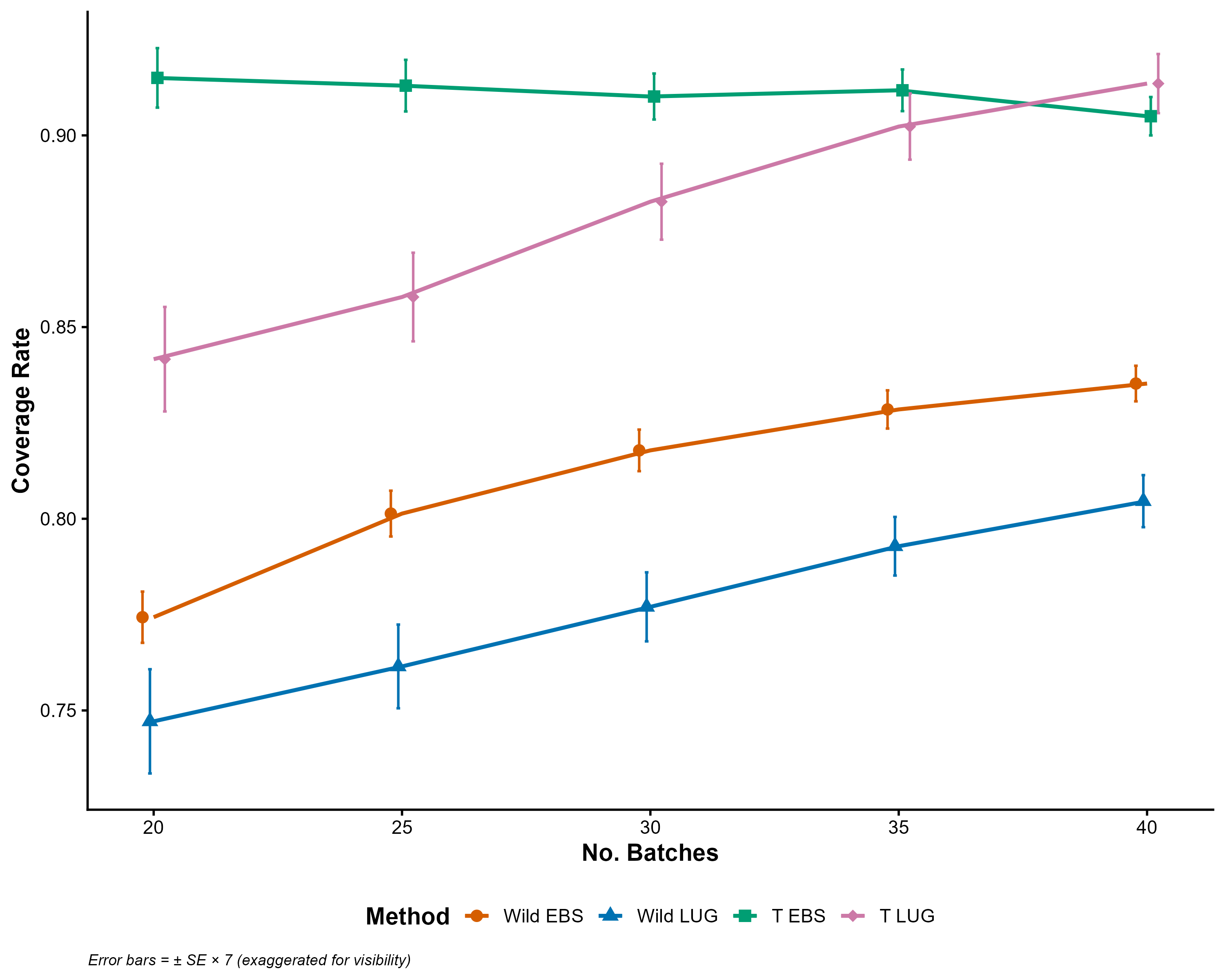}
		\caption{Dimension 30, number of batches 20 to 40.}
		\label{covg_30_batch20_40}
	\end{subfigure}
	\hfill
	\begin{subfigure}[htbp!]{0.32\linewidth}
		\centering
		\includegraphics[width=\linewidth]{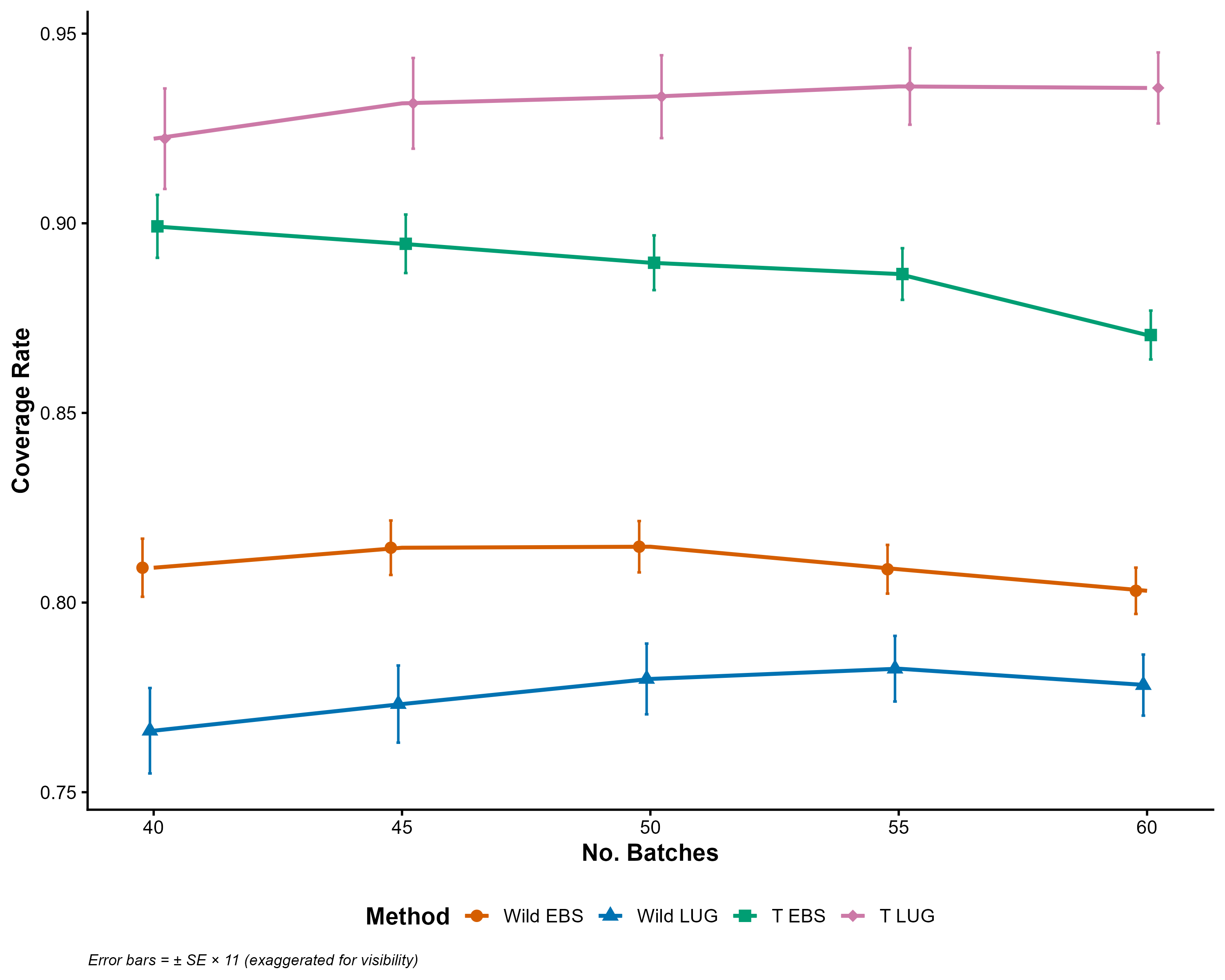}
		\caption{Dimension 50, number of batches 40 to 60.}
		\label{covg_50_batch40_60}
	\end{subfigure}
	
	\vspace{0.5cm}
	
	\begin{subfigure}[htbp!]{0.32\linewidth}
		\centering
		\includegraphics[width=\linewidth]{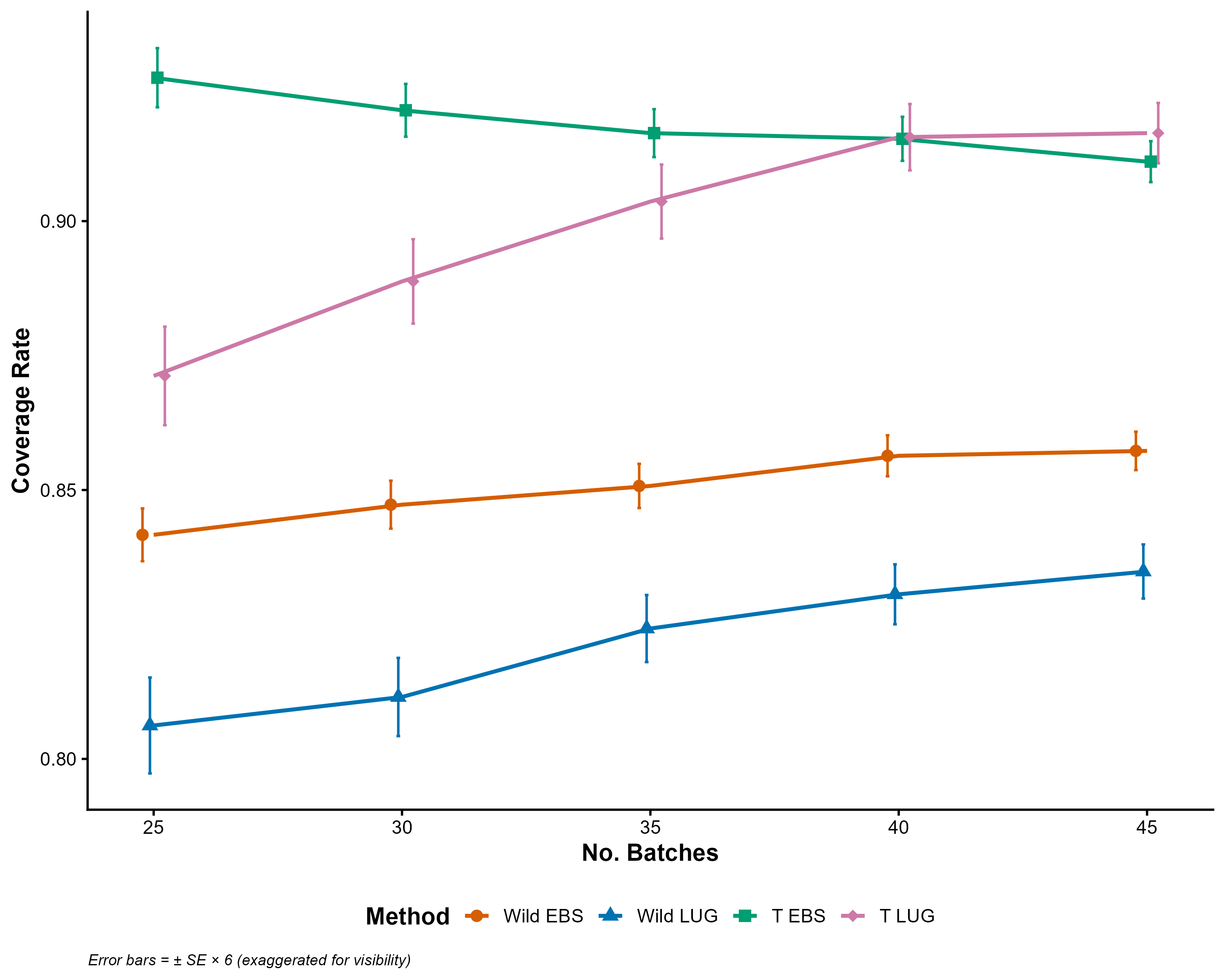}
		\caption{Dimension 20, number of batches 25 to 45.}
		\label{covg_20_batch25_45}
	\end{subfigure}
	\hfill
	\begin{subfigure}[htbp!]{0.32\linewidth}
		\centering
		\includegraphics[width=\linewidth]{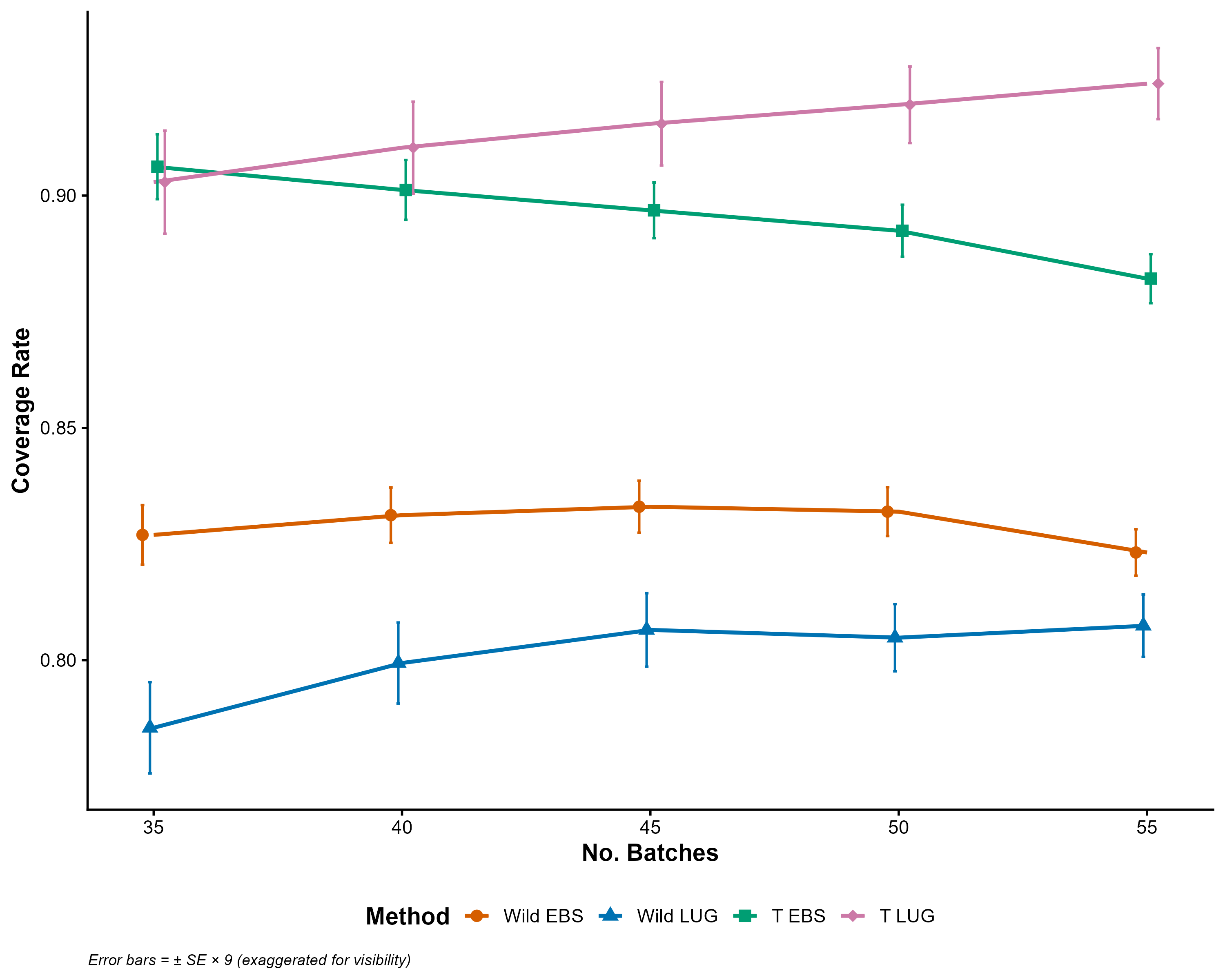}
		\caption{Dimension 30, number of batches 35 to 55.}
		\label{covg_30_batch35_55}
	\end{subfigure}
	\hfill
	\begin{subfigure}[htbp!]{0.32\linewidth}
		\centering
		\includegraphics[width=\linewidth]{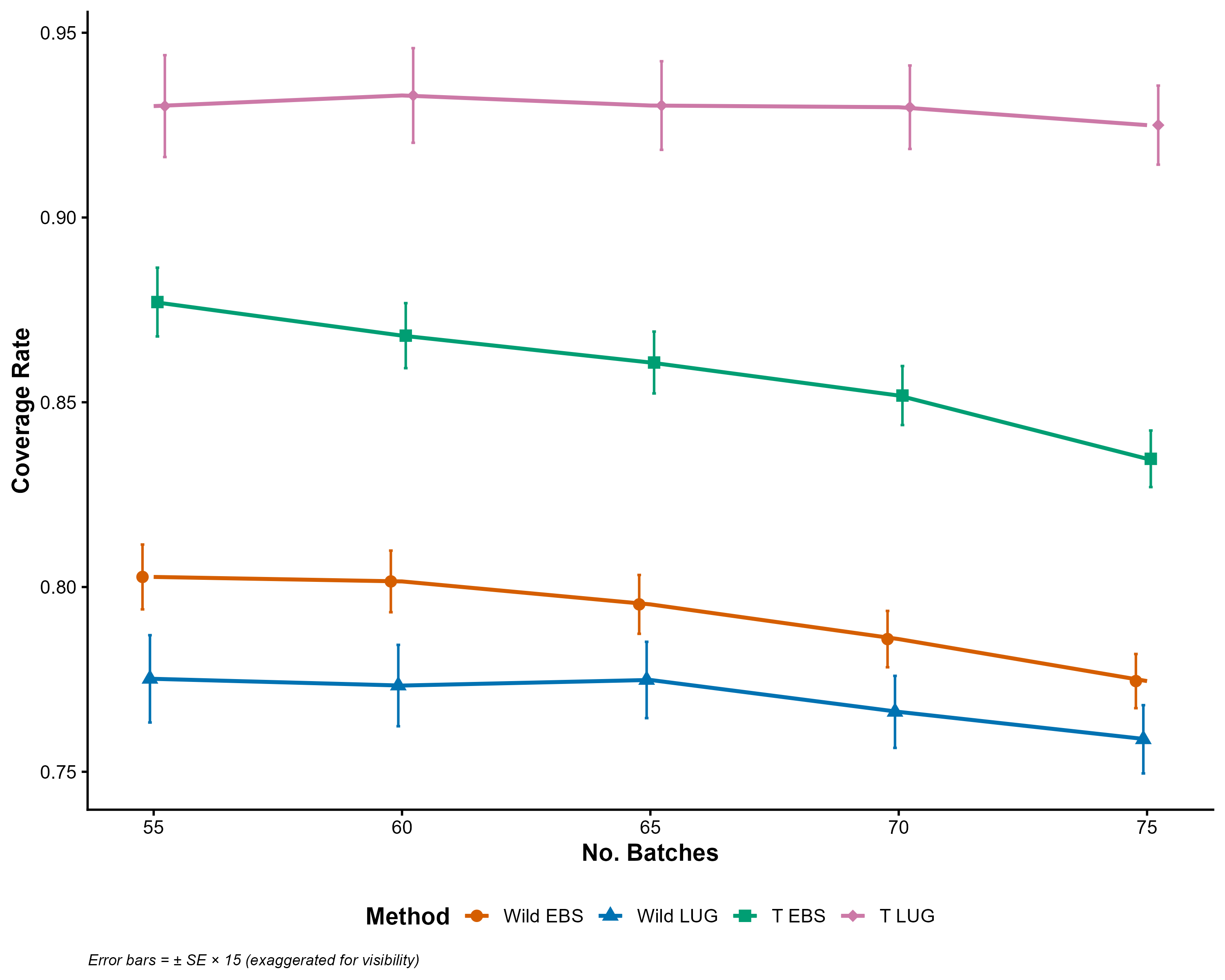}
		\caption{Dimension 50, number of batches 55 to 75.}
		\label{covg_50_batch55_75}
	\end{subfigure}
	
	\caption{{Multivariate coverage based on simultaneous marginal friendly implementation, for varying dimension and number of batches for fixed sample size problem.}}
	\label{covg_vary_batch}
\end{figure}

\begin{figure}[htbp!]
	\centering
	
	\begin{subfigure}[htbp!]{0.32\linewidth}
		\centering
		\includegraphics[width=\linewidth]{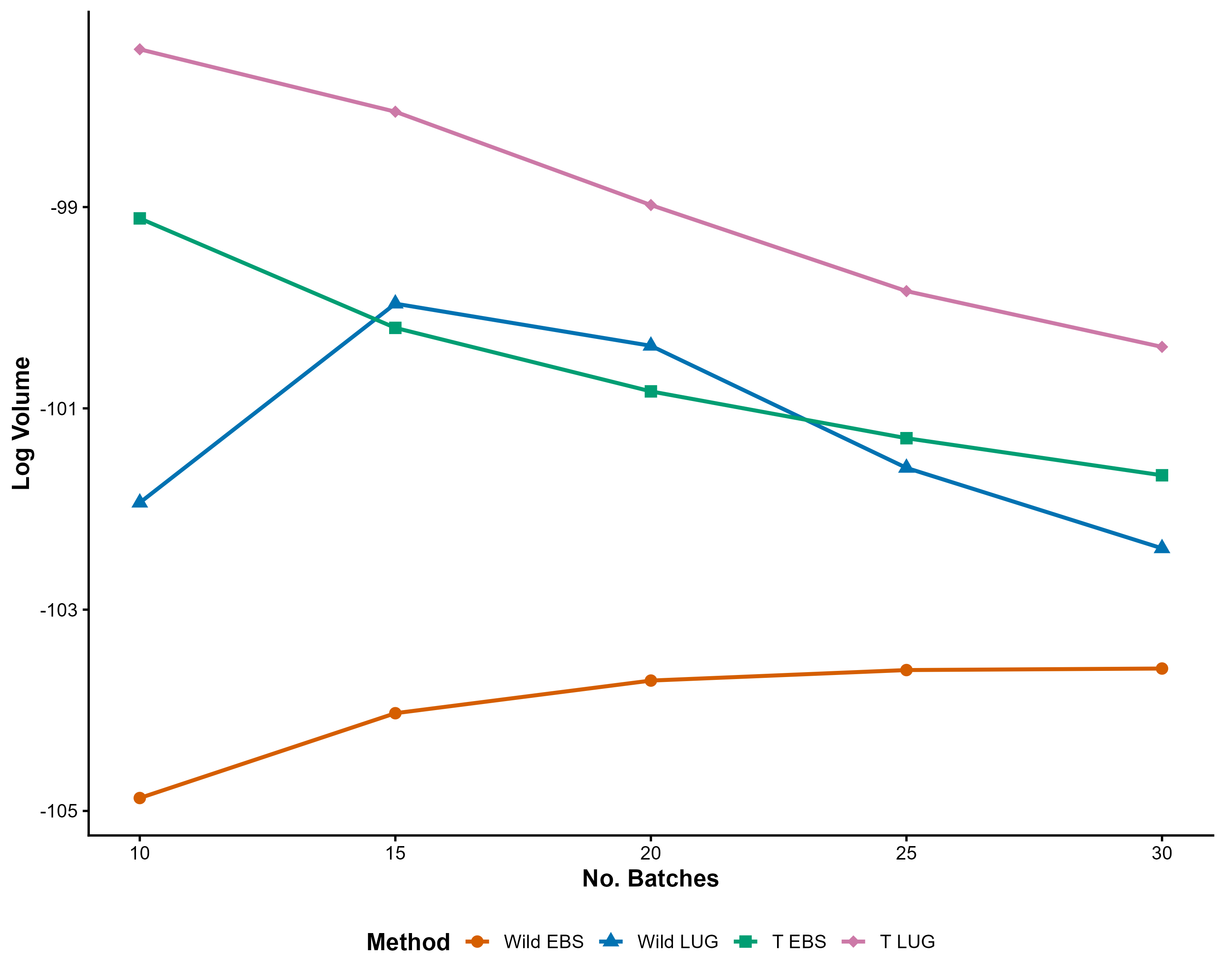}
		\caption{Dimension 20, number of batches 10 to 30.}
		\label{vol_rect_20_batch10_30}
	\end{subfigure}
	\hfill
	\begin{subfigure}[htbp!]{0.32\linewidth}
		\centering
		\includegraphics[width=\linewidth]{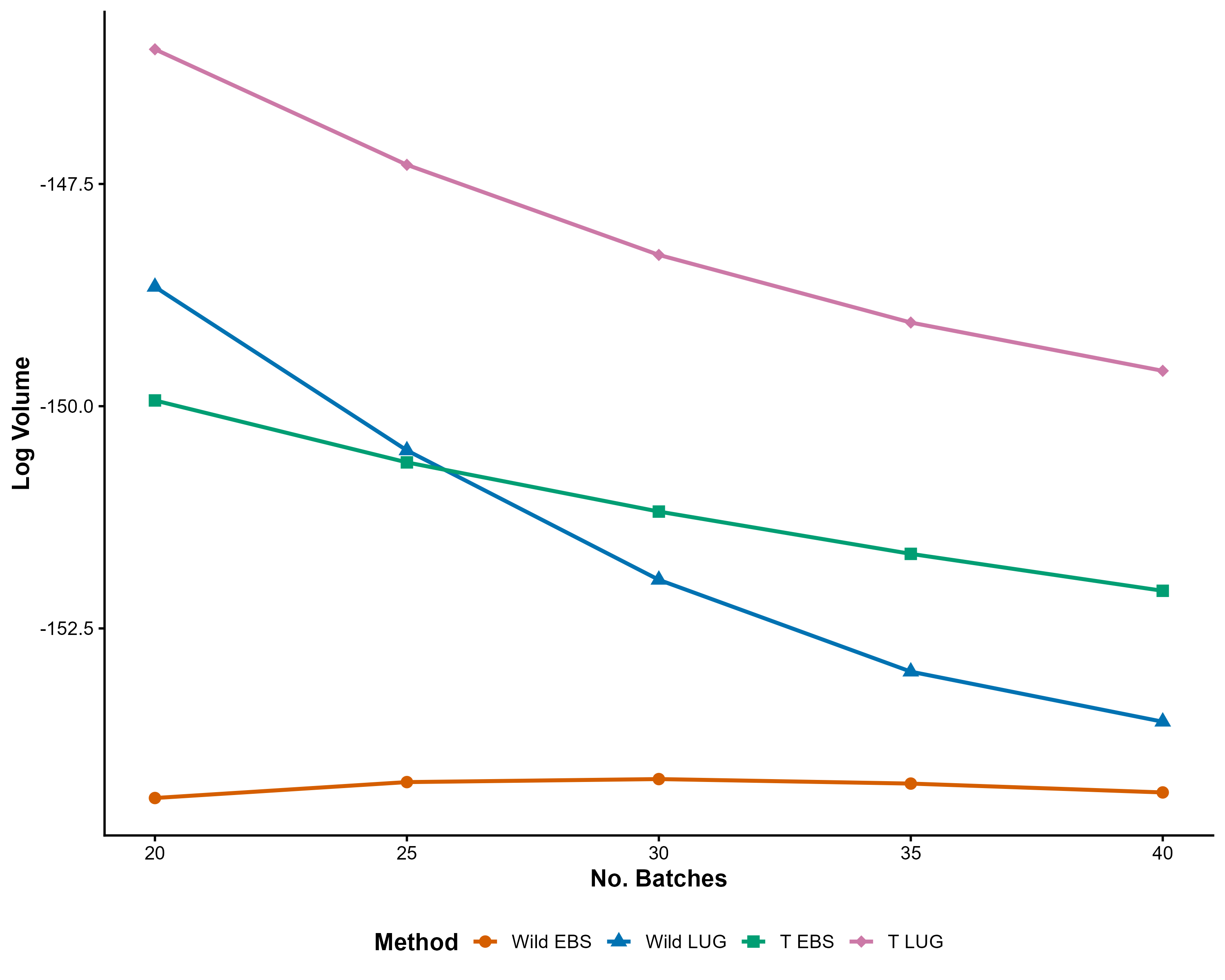}
		\caption{Dimension 30, number of batches 20 to 40.}
		\label{vol_rect_30_batch20_40}
	\end{subfigure}
	\hfill
	\begin{subfigure}[htbp!]{0.32\linewidth}
		\centering
		\includegraphics[width=\linewidth]{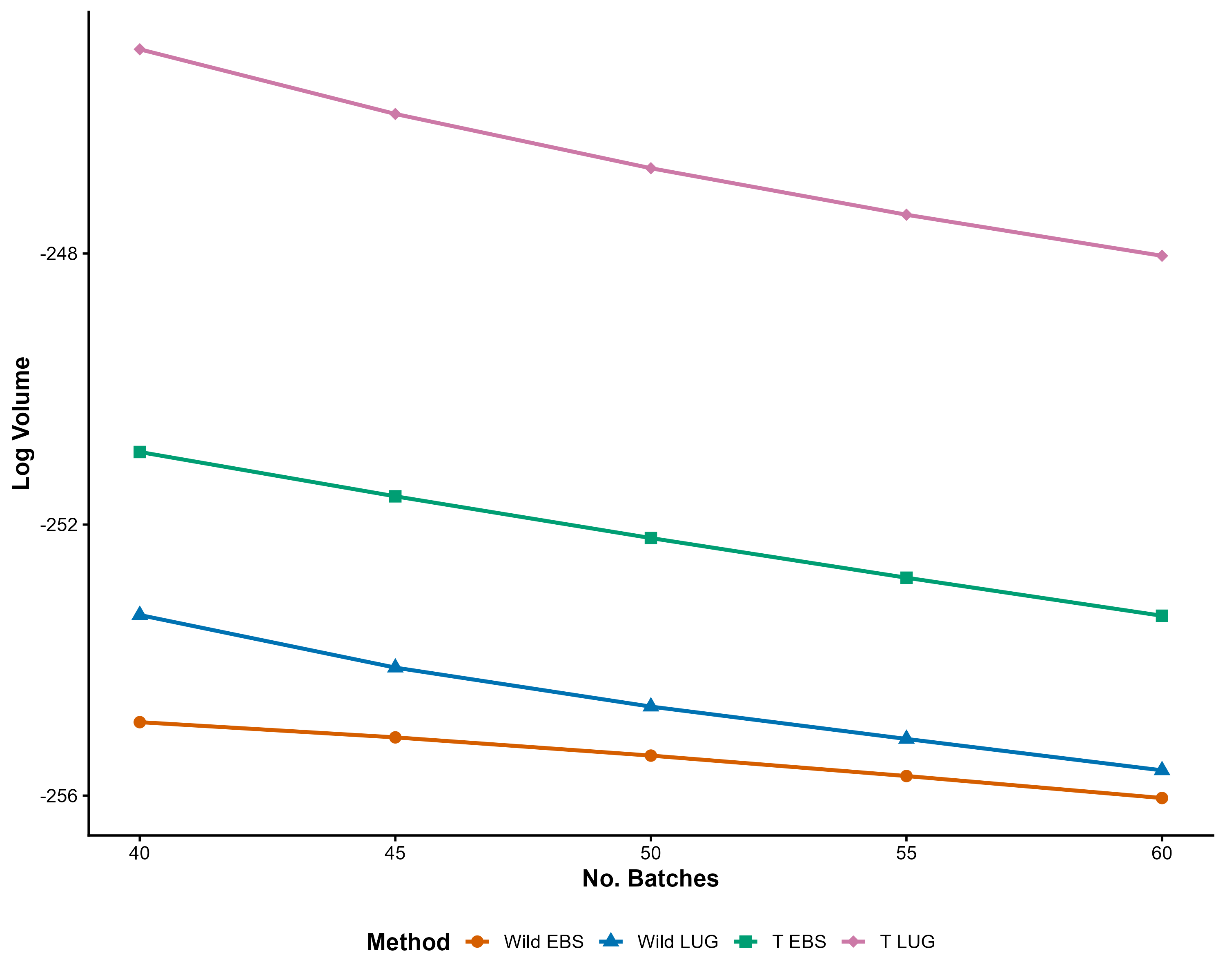}
		\caption{Dimension 50, number of batches 40 to 60.}
		\label{vol_rect_50_batch40_60}
	\end{subfigure}
	
	\vspace{0.5cm}
	
	\begin{subfigure}[htbp!]{0.32\linewidth}
		\centering
		\includegraphics[width=\linewidth]{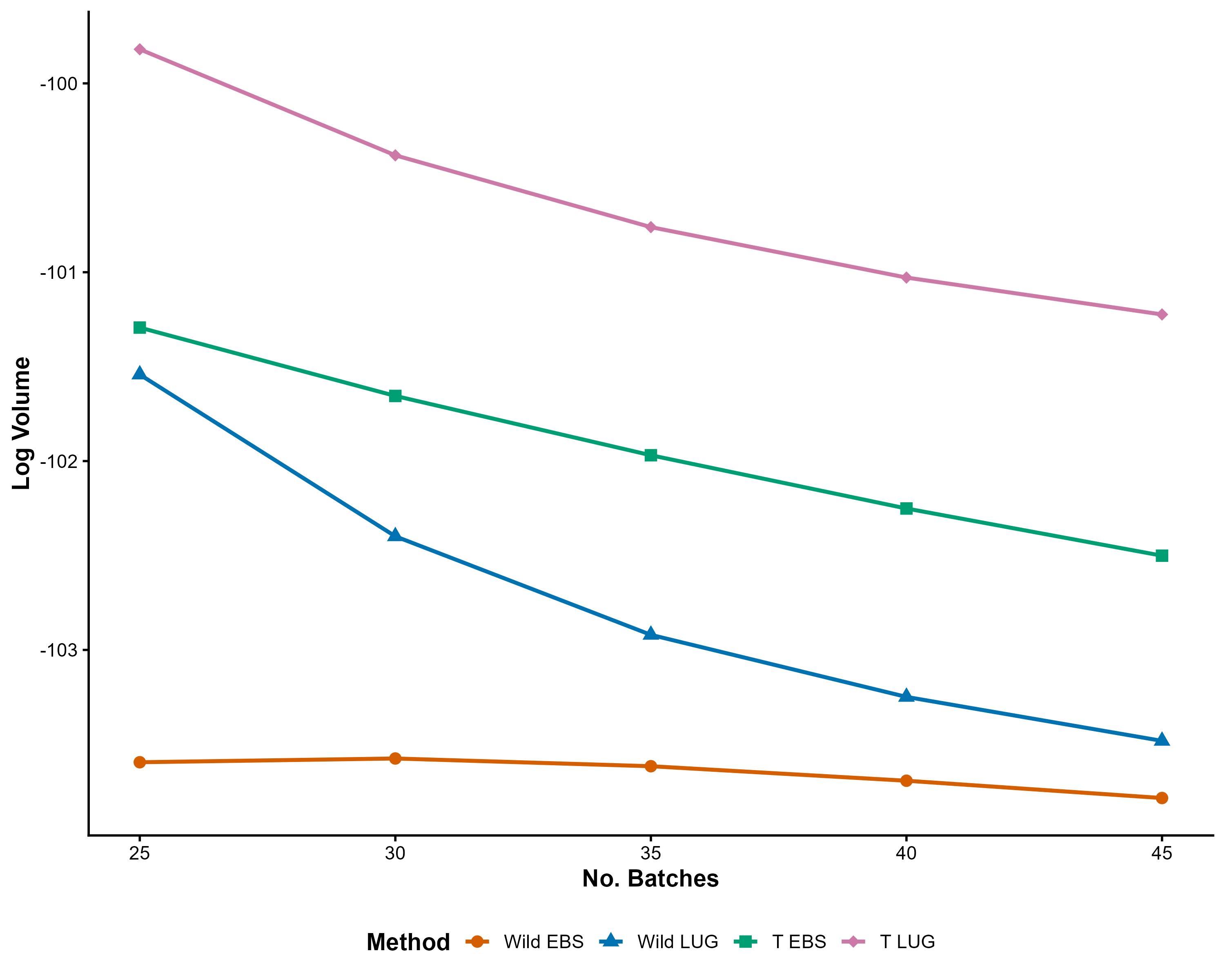}
		\caption{Dimension 20, number of batches 25 to 45.}
		\label{vol_rect_20_batch25_45}
	\end{subfigure}
	\hfill
	\begin{subfigure}[htbp!]{0.32\linewidth}
		\centering
		\includegraphics[width=\linewidth]{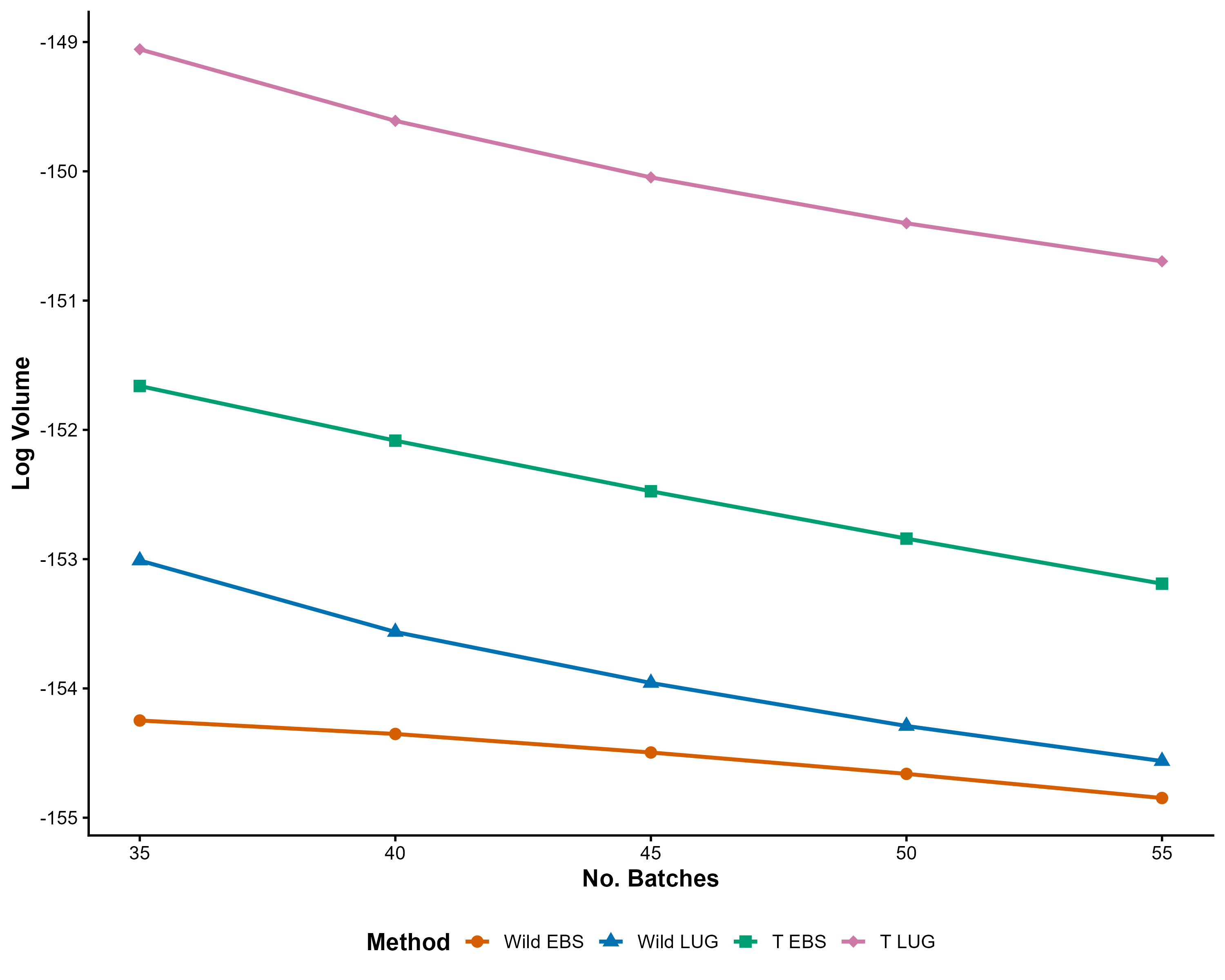}
		\caption{Dimension 30, number of batches 35 to 55.}
		\label{vol_rect_30_batch35_55}
	\end{subfigure}
	\hfill
	\begin{subfigure}[htbp!]{0.32\linewidth}
		\centering
		\includegraphics[width=\linewidth]{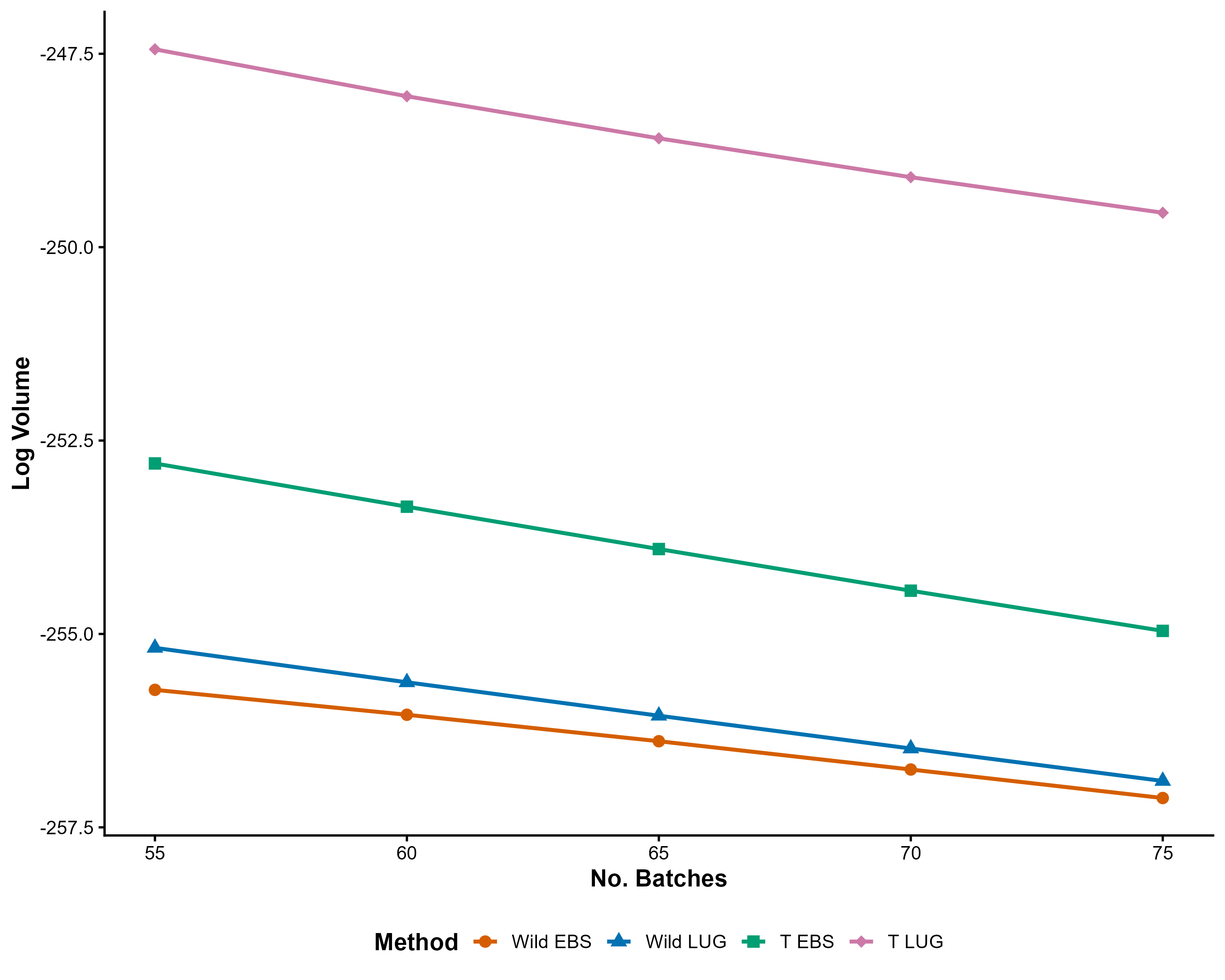}
		\caption{Dimension 50, number of batches 55 to 75.}
		\label{vol_rect_50_batch55_75}
	\end{subfigure}
	
	\caption{{Volume (logarithm) of confidence hyper-rectangle based on simultaneous marginal friendly implementation, for varying dimension and number of batches for fixed sample size problem.
    }}
	\label{vol_rect_vary_batch}
\end{figure}

\FloatBarrier

\subsection{Comparison with the Bonferroni Correction}
Figure \ref{Bonf_comp} compares the proposed simultaneous regions against a baseline Bonferroni correction applied to the marginal intervals. When the covariates possess an independent structure, the Bonferroni method performs optimally at large sample sizes, which is expected. However, under dependent covariate structures (such as Toeplitz or equi-correlation), the Bonferroni approach suffers from degraded multivariate coverage. Because aggregated marginal length metrics fail to adequately capture joint estimation dynamics in these dependent settings, we evaluate performance strictly via the overall volume of the confidence hyper-rectangles. We observe that the Lugsail $t$-copula implementation (T LUG) successfully recovers the best multivariate coverage rates, achieving this by generating the largest hyper-rectangular volumes across the evaluated sample sizes, with the Bonferroni method yielding the second largest. Consistent with prior results, the volumes of the Lugsail-corrected regions decay more slowly than their standard EBS counterparts. Ultimately, this conservative volume retention by the Lugsail $t$-copula is structurally necessary to robustly capture the finite-sample dependencies and outperform the Bonferroni correction in highly correlated designs.

\begin{figure}[htbp!]
	\centering
	
	\begin{subfigure}[htbp!]{0.32\linewidth}
		\centering
		\includegraphics[width=\linewidth]{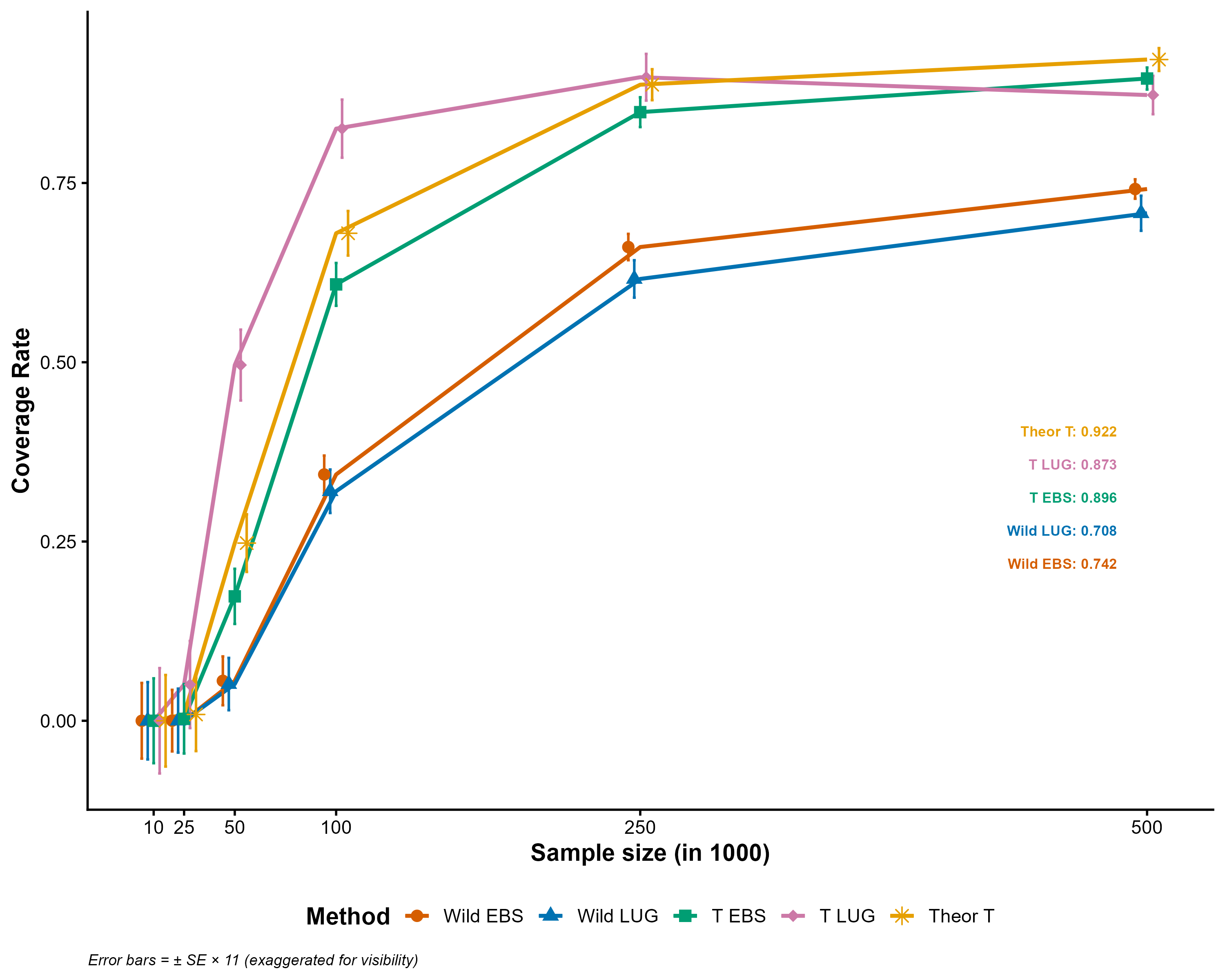}
		\caption{Multivariate coverage in independent case of $x$.}
		\label{covg_Bonf_indep}
	\end{subfigure}
	\hfill
	\begin{subfigure}[htbp!]{0.32\linewidth}
		\centering
		\includegraphics[width=\linewidth]{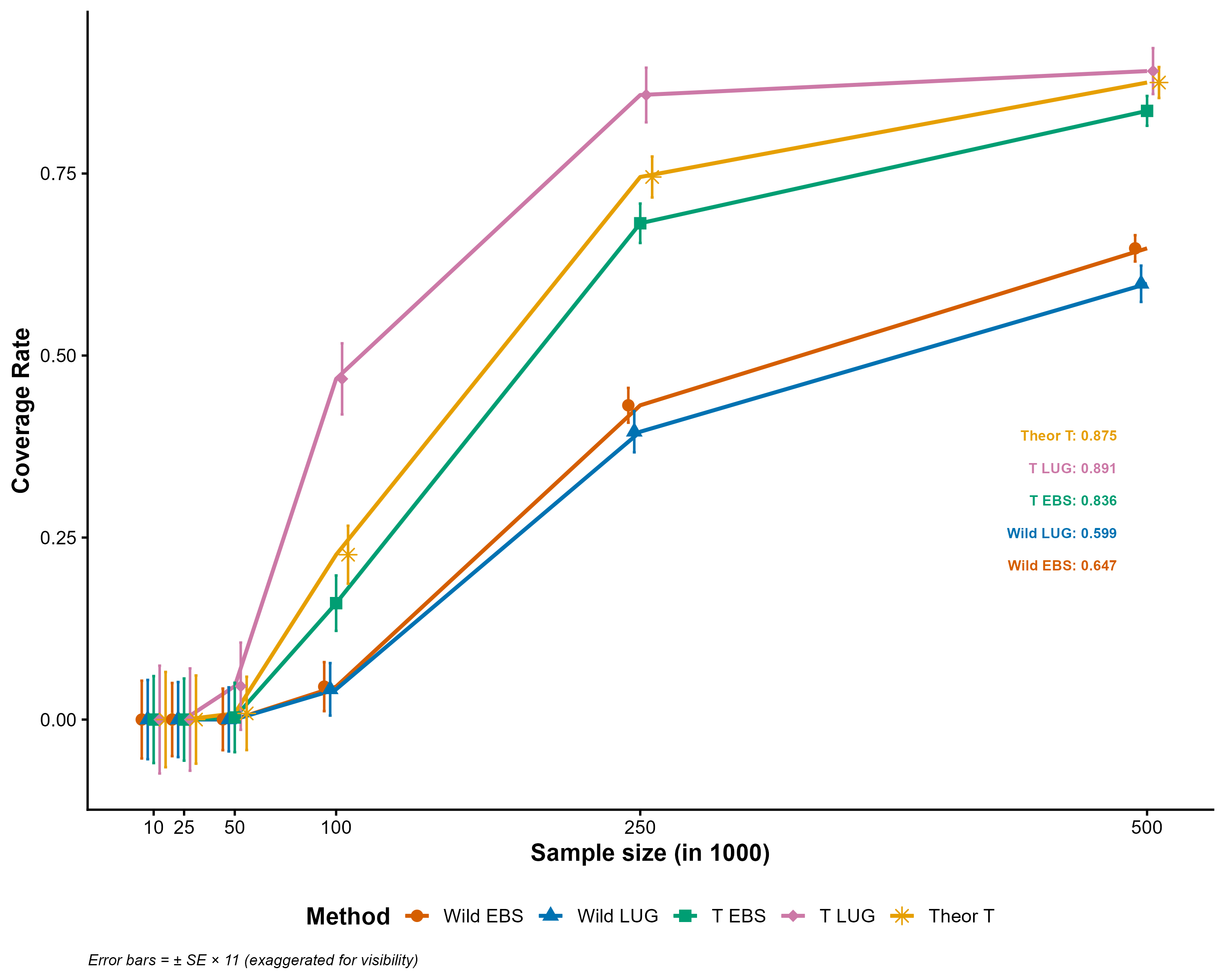}
		\caption{Multivariate coverage in equi-correlation case of $x$.}
		\label{covg_Bonf_equiv}
	\end{subfigure}
	\hfill
	\begin{subfigure}[htbp!]{0.32\linewidth}
		\centering
		\includegraphics[width=\linewidth]{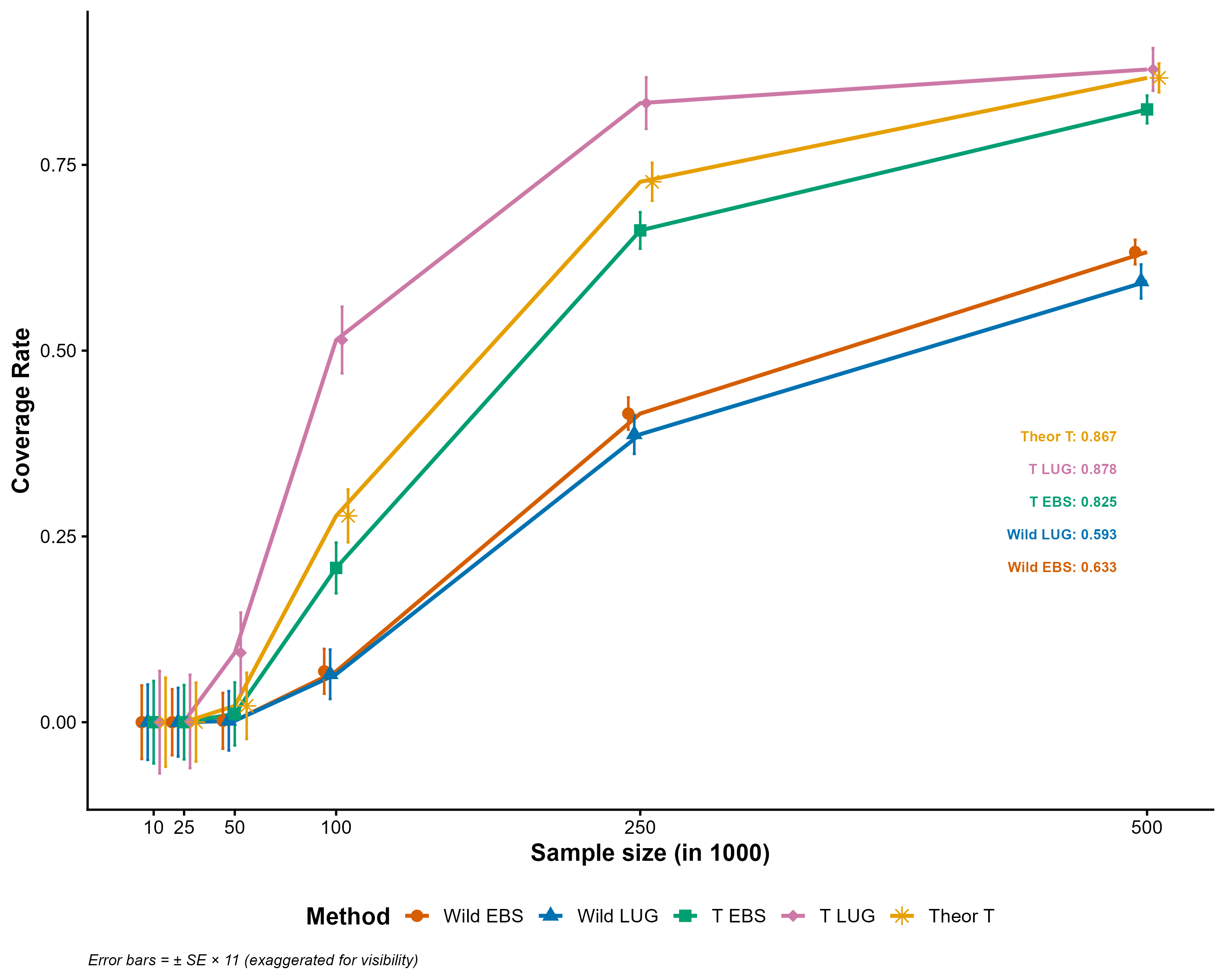}
		\caption{Multivariate coverage in Toeplitz case of $x$.}
		\label{covg_Bonf_toep}
	\end{subfigure}
	
	\vspace{0.5cm}
	
	\begin{subfigure}[htbp!]{0.32\linewidth}
		\centering
		\includegraphics[width=\linewidth]{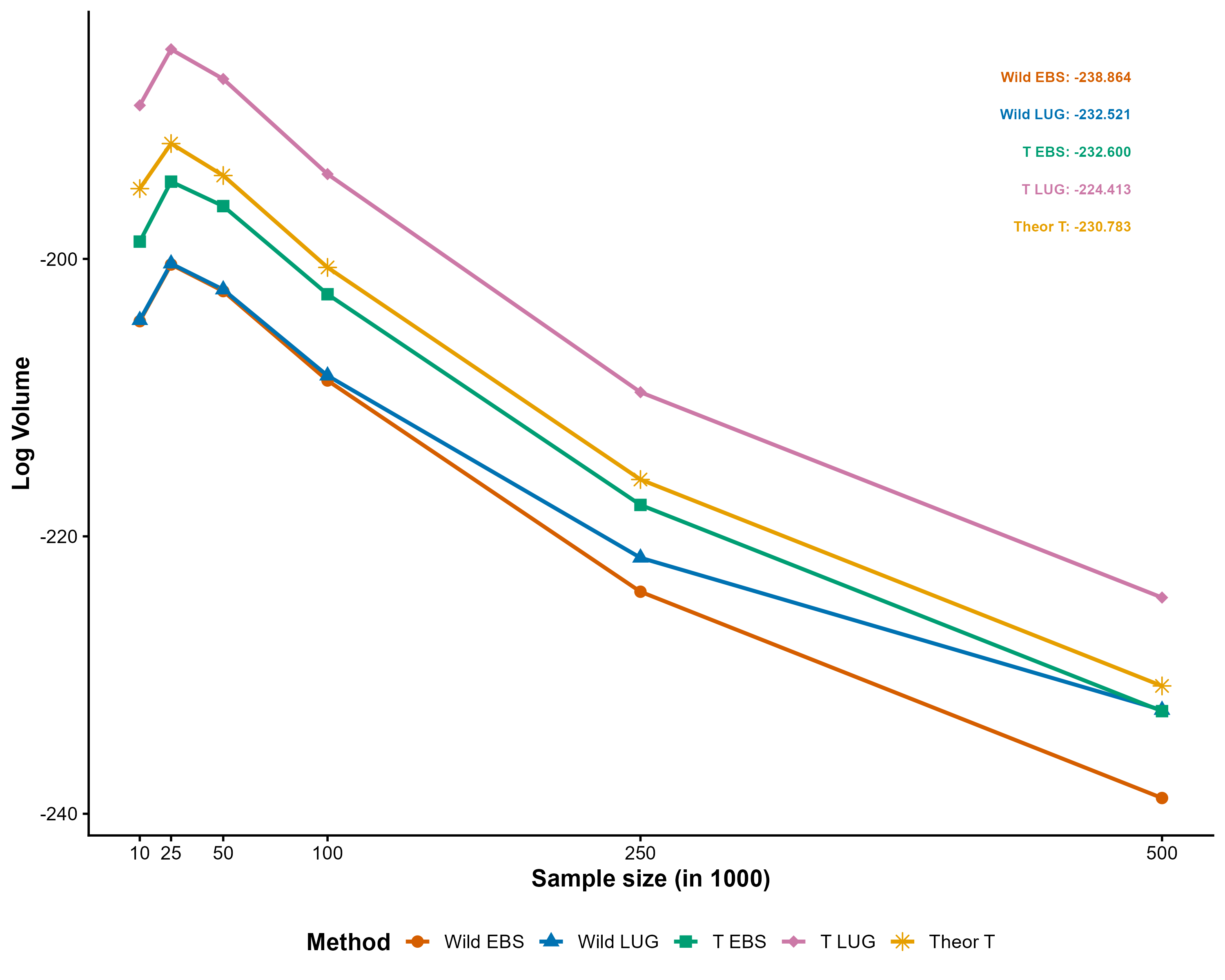}
		\caption{Volume (logarithm) of hyper-rectangle in independent case of $x$.}
		\label{vol_rect_Bonf_indep}
	\end{subfigure}
	\hfill
	\begin{subfigure}[htbp!]{0.32\linewidth}
		\centering
		\includegraphics[width=\linewidth]{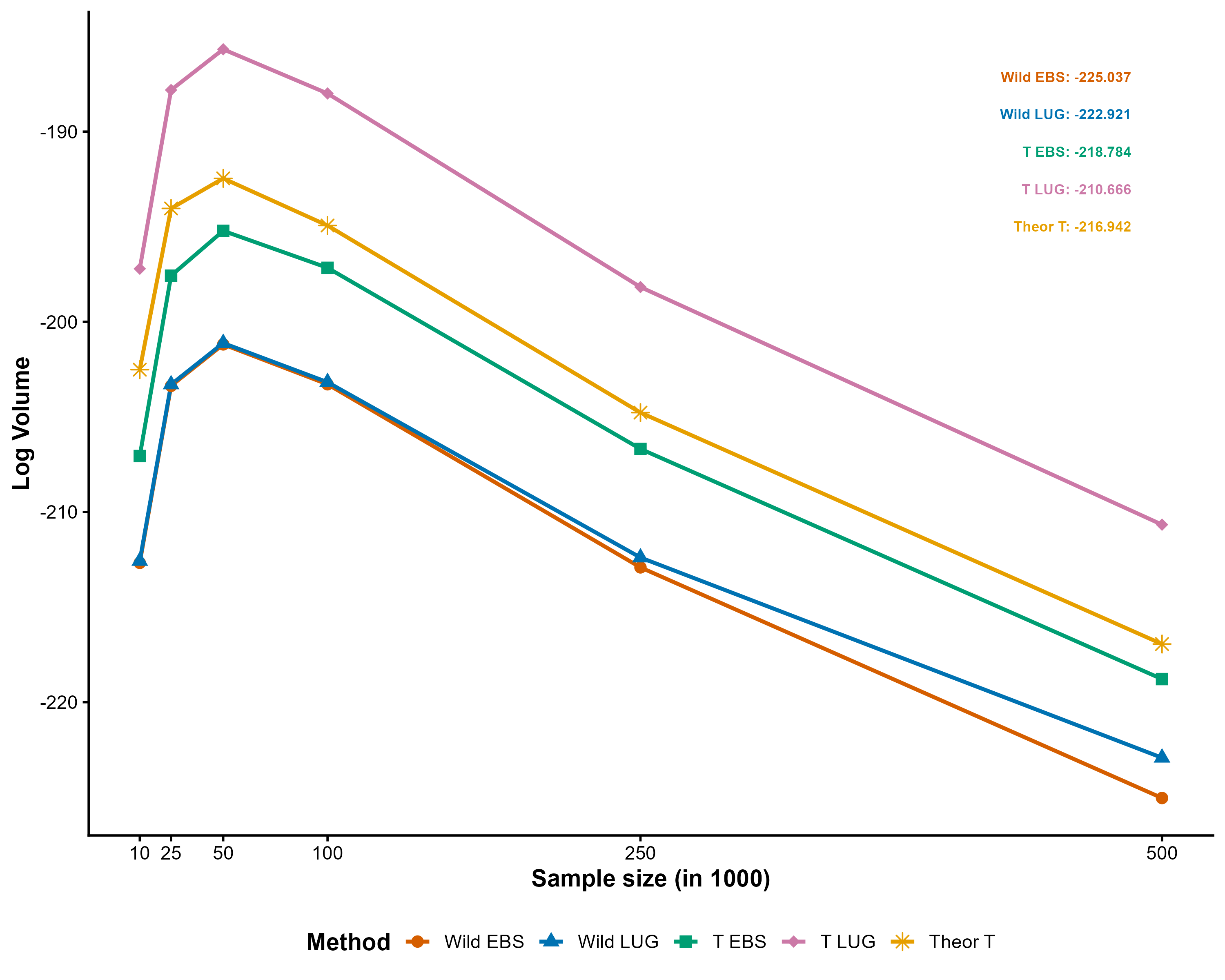}
		\caption{Volume (logarithm) of hyper-rectangle in equi-correlation case of $x$.}
		\label{vol_rect_Bonf_equiv}
	\end{subfigure}
	\hfill
	\begin{subfigure}[htbp!]{0.32\linewidth}
		\centering
		\includegraphics[width=\linewidth]{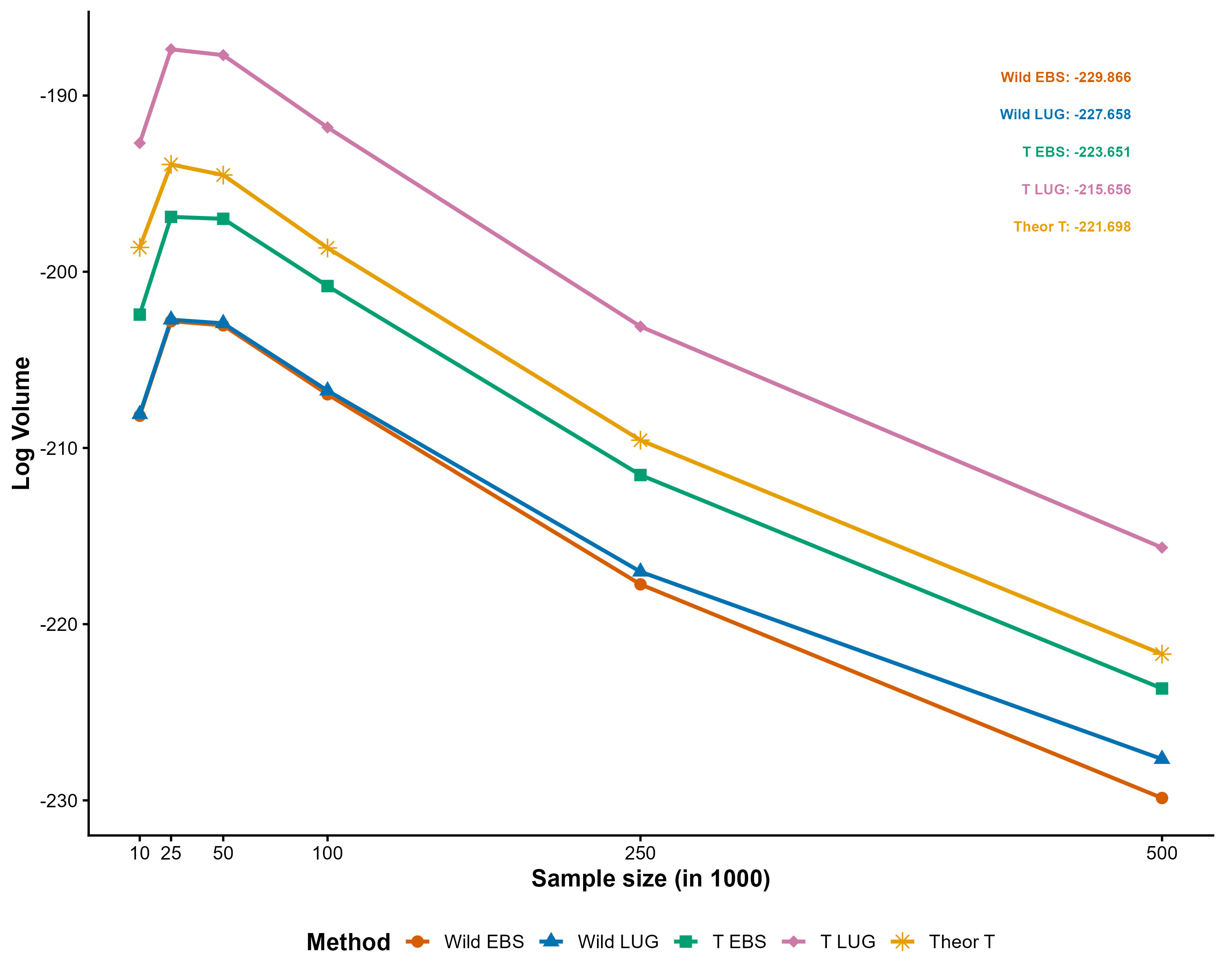}
		\caption{Volume (logarithm) of hyper-rectangle in Toeplitz case of $x$.}
		\label{vol_rect_Bonf_toep}
	\end{subfigure}
	
	\caption{{Comparing Bonferroni correction with the marginal friendly hyper-rectangles on the basis of coverage and volume contained (logarithm) based on different cases of variance-covariance matrix of the regressor in linear model. The dimension of parameters is 50 and number of batches is 25. }}
	\label{Bonf_comp}
\end{figure}

\subsection{Heavy-Tailed Errors}
Finally, we evaluate the robustness of the proposed framework under heavy-tailed error distributions, specifying a configuration of $d=50$ and $m=25$ with an equicorrelation covariance structure. Assuming standard Laplace distributed errors, the empirical coverage rates and hyper-rectangular volumes for the Least Absolute Deviations (LAD) estimates are illustrated in Figure \ref{fig:heavy_tailed_results}. For the $t$-copula procedure, integrating the Lugsail bias correction yields a substantial improvement in multivariate coverage over the standard EBS estimator (Figure \ref{fig:LAD_covg_dim50_t_err}), albeit at the expected cost of an increased confidence volume (Figure \ref{fig:LAD_vol_dim50}).

\begin{figure}[htbp!]
    \centering

    \begin{subfigure}[htbp!]{0.48\linewidth}
        \centering
        \includegraphics[width=\linewidth]{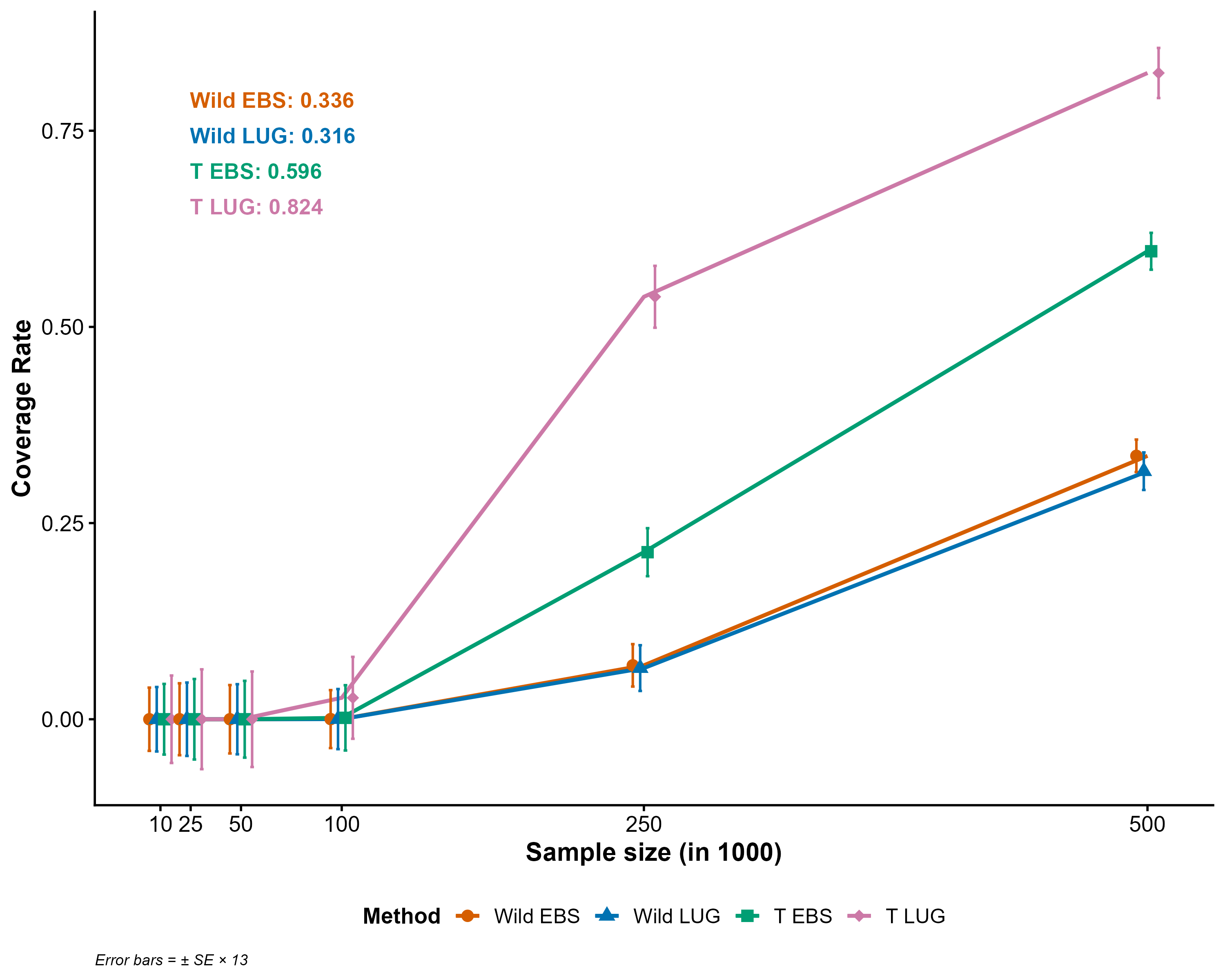}
        \caption{Multivariate coverage for LAD estimate.}
        \label{fig:LAD_covg_dim50_t_err}
    \end{subfigure}
    \hfill
    \begin{subfigure}[htbp!]{0.48\linewidth}
        \centering
        \includegraphics[width=\linewidth]{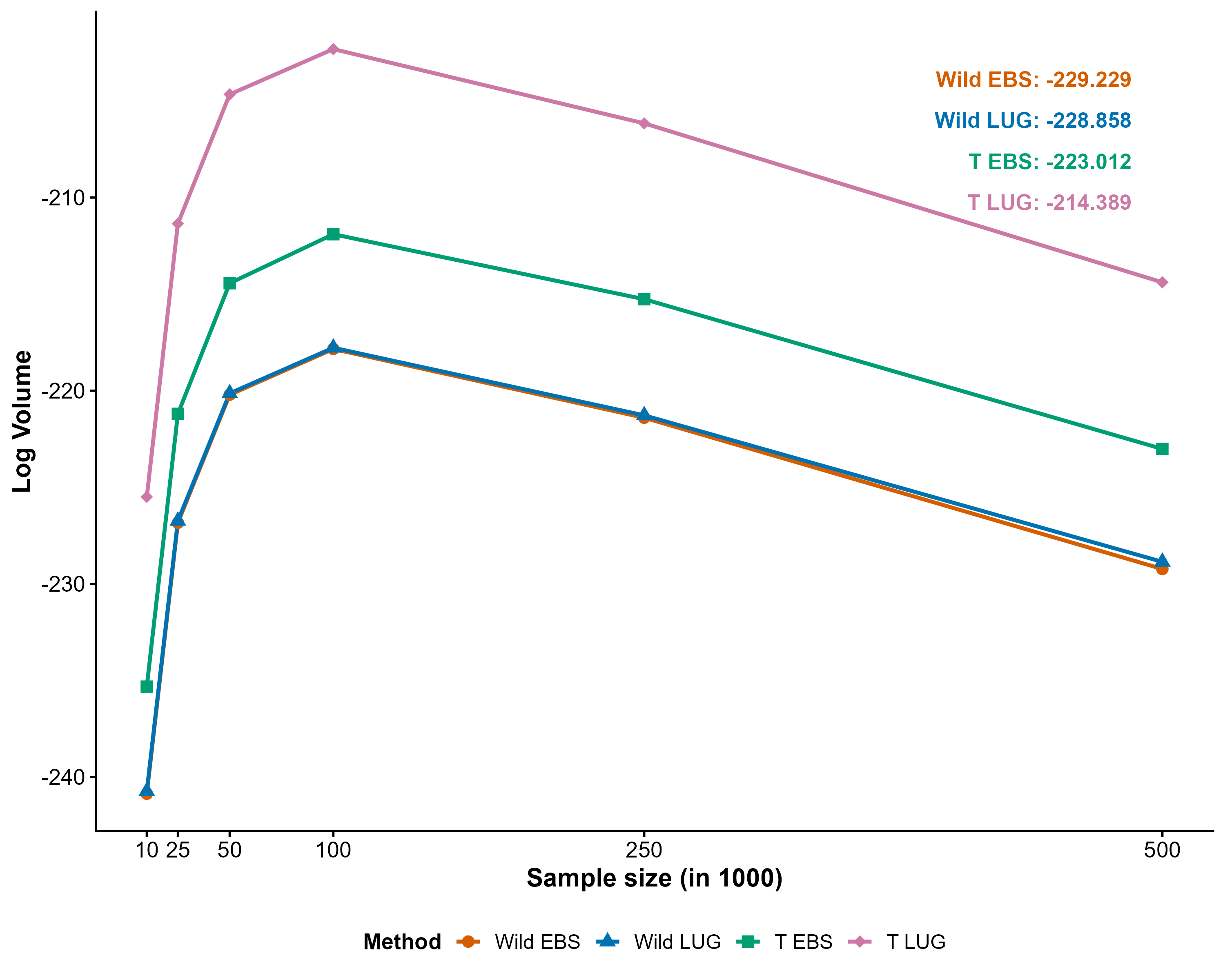}
        \caption{Volume of confidence hyper-rectangle for LAD estimate.}
        \label{fig:LAD_vol_dim50}
    \end{subfigure}
    
    \caption{Performance of the proposed estimators for  $d=50$ with $m=25$ and an equicorrelation matrix structure, evaluating Least Absolute Deviations (LAD) estimates under standard Laplace errors.}
    \label{fig:heavy_tailed_results}
\end{figure}

The simulation study demonstrates the empirical efficacy of the proposed framework, yielding three primary conclusions. First, the marginal formulation successfully bypasses the matrix inversion bottleneck of prior methods, delivering valid simultaneous confidence regions even in highly constrained regimes where the number of batches is less than or equal to the parameter dimension ($m \le d$). Second, the $t$-copula approximation consistently outperforms the Wild Bootstrap across all metrics by explicitly mapping the underlying empirical cross-covariance structure. Third, the choice between the standard EBS estimator and the Lugsail bias correction reveals a strict bias-variance tradeoff governed by the batch-to-dimension ratio. While the Lugsail correction effectively improves finite-sample coverage, particularly for heavy-tailed LAD models, its requirement to collapse adjacent batches significantly inflates the hyperrectangular volume. Consequently, the standard EBS estimator is optimal in low-batch settings ($m < d$) and large-sample asymptotes, whereas the Lugsail correction is advantageous only when batches are abundant (e.g., $m \ge 2d$). Ultimately, all proposed estimators correctly exhibit the theoretically expected asymptotic decay in volume and interval length as the sample size $n$ increases.

\section{Conclusion}
\label{sec:conclusion}
By integrating the EBS marginal-friendly projection and Lugsail bias-correction techniques of \citet{singh2025} with the cancellation method of \citet{zhu2021}, we derive a robust inference framework for SGD which  uplifts the limitation of number of batches being more than the dimension of the problem.  
The mathematical separation of the marginal scale factor completely removes the necessity to invert the full $d \times d$ covariance matrix, entirely resolving the $m > d$ degeneracy problem \citep{zhu2021}. The main contribution is to obtain marginal friendly quantiles based on two sampling procedures, wild bootstrap from distribution of the considered statistic and  iid samples from $t$-copula with a fitted correlation matrix. The second major contribution is to obtain finite sample improvement in marginal-friendly inference based on lugsail adjustment of the batch means estimator. 

The efficacy of the Lugsail bias correction is highly dependent on the available degrees of freedom. When $m < d$, the standard Equal Batch-Size (EBS) estimator dominates, whereas the Lugsail versions produce wider, slower-decaying intervals (with the Lugsail $t$-copula being the widest). This under-performance is fundamentally driven by Lugsail's requirement to collapse adjacent batches, halving an already small $m$. Conversely, when batches are abundant ($m > 2d$), the variance stabilizes, and the Lugsail correction successfully outperforms the standard EBS estimator. Moreover, both of the proposed algorithms require solving a 1-dimensional search problem in order to find the quantile of the distribution based on the bootstrap or iid samples. 

Several promising avenues remain for future research. First, to overcome the variance inflation associated with the Lugsail estimator in high dimensions, future work could integrate Overlapping Batch Means (OBM) or shrinkage-based regularization to stabilize the covariance estimation without sacrificing effective degrees of freedom. Second, while our theoretical guarantees rely on the strong convexity of the objective function, extending the marginal cancellation framework to non-convex optimization landscapes, such as deep neural networks, is a critical next step. Finally, adapting this simultaneous inference methodology to distributed or federated learning environments—where communication bottlenecks strictly limit the number of feasible batch computations—presents an important practical challenge for scalable statistical inference.

\section*{Acknowledgments}
 The work of Rahul Singh was supported by New Faculty Seed Grant No. MI03038G at Indian Institute of Technology Delhi, India.

\bigskip
\appendix
\section{Appendix}

\subsection{Proof of Theorem \ref{thm1}} \label{appendix1}
\begin{proof} 
The notation `$\Rightarrow$' denotes convergence in distribution. By applying the continuous mapping theorem to Theorem 2 in \citet{zhu2021}, as $n\to\infty$, we have
\begin{equation*}
\sqrt{n} \sqrt{m} (\hat{\theta}_{n,j} - \theta^*_j) \Rightarrow \sqrt{m} e_j^\top G B(1). 
\end{equation*}
where $G G^\top = \Sigma$, $e_j$ is the standard basis vector for the $j$-th dimension and $B(\cdot)$ is a $d$-dimensional standard Brownian motion. 
Let $v_j = G^\top e_j$, then $\|v_j\|^2 = e_j^\top G G^\top e_j = e_j^\top \Sigma e_j = \Sigma_{jj}$. For all $t \in [0,1]$, the process $v_j^\top B(t)$ is a linear combination of independent Brownian motions, forming a scalar Brownian motion with variance $\Sigma_{jj} t$, denoted by $W_j(t)$, such that $v_j^\top B(t) = \sqrt{\Sigma_{jj}} W_j(t)$. 
Evaluating this process at the endpoint $t=1$, the asymptotic limit of the scaled numerator becomes $\sqrt{m} \sqrt{\Sigma_{jj}} W_j(1)$.

Next, again using Theorem 2 in \cite{zhu2021}, as $n\to\infty$, we have
\begin{equation*}
n S_m(n) \Rightarrow \frac{1}{m-1} \sum_{i=1}^m G(m \Delta B_i - B(1))(m \Delta B_i - B(1))^\top G^\top 
\end{equation*}
where $\Delta B_i = B(i/m) - B((i-1)/m)$ is the increment of the $d$-dimensional Brownian motion over the $i$-th batch. The $j$-th diagonal element of $n S_m(n)$ is
\begin{equation*}
n [S_m(n)]_{jj} \Rightarrow \frac{1}{m-1} \sum_{i=1}^m \left(e_j^\top G (m \Delta B_i - B(1))\right)^2
\end{equation*}
By substituting $e_j^\top G = v_j^\top$ and defining $\Delta W_{j,i} = W_j(i/m) - W_j((i-1)/m)$, the expression simplifies to
\begin{equation*}
n [S_m(n)]_{jj} \Rightarrow \frac{1}{m-1} \sum_{i=1}^m \left(\sqrt{\Sigma_{jj}} (m \Delta W_{j,i} - W_j(1))\right)^2
\end{equation*}

Now, some straightforward algebra yields
\begin{equation}\label{pf:eq1}
t_{n,j} \Rightarrow \frac{\sqrt{m} W_j(1)}{\sqrt{\frac{1}{m-1} \sum_{i=1}^m \left( m \Delta W_{j,i} - W_j(1) \right)^2}}
\end{equation}

To determine the exact distribution of this ratio, we map the Brownian variables to standard normal random variables. By the properties of standard Brownian motion, the non-overlapping increments are independent and distributed as $\Delta W_{j,i} \sim \mathcal{N}(0, 1/m)$. 
We define a set of iid standard normal random variables $Z_i = \sqrt{m} \Delta W_{j,i}$, such that $Z_i \sim \mathcal{N}(0, 1)$. Notice that,
\begin{equation*}
W_j(1) = \sum_{i=1}^m \Delta W_{j,i} = \sum_{i=1}^m \frac{Z_i}{\sqrt{m}} = \sqrt{m} \left( \frac{1}{m} \sum_{i=1}^m Z_i \right) = \sqrt{m} \bar{Z}
\end{equation*}
where $\bar{Z}= \sum_{i=1}^m Z_i/m$. Furthermore, $m \Delta W_{j,i} - W_j(1) = m \left( \frac{Z_i}{\sqrt{m}} \right) - \sqrt{m} \bar{Z} = \sqrt{m}(Z_i - \bar{Z})$, implying that
\begin{equation*}
{\sqrt{\frac{1}{m-1} \sum_{i=1}^m \left( m \Delta W_{j,i} - W_j(1) \right)^2}}= 
\sqrt{m} \sqrt{\frac{\sum_{i=1}^m (Z_i - \bar{Z})^2}{m-1}} = \sqrt{m} S_Z
\end{equation*}
where $S_Z$ is the sample standard deviation of $Z_i$'s. Consequently, \eqref{pf:eq1} simplifies to
\begin{equation*}
t_{n,j} \Rightarrow \frac{m \bar{Z}}{\sqrt{m} S_Z} = \frac{\sqrt{m} \bar{Z}}{S_Z}
\end{equation*}
Therefore, by definition of the Student t-distribution, for any $m\ge 2$, $t_{n,j}$ converges in distribution to Student's t-distribution with $m-1$ degrees of freedom. 
\end{proof}

\end{document}